\documentclass[12pt,letterpaper]{scrartcl}
\usepackage{amsmath}
\usepackage{amsthm}
\usepackage{amssymb}
\usepackage{stmaryrd}
\usepackage{accents}
\usepackage{bbm}
\usepackage[mathscr]{euscript}
\usepackage{xcolor}
\definecolor{Myblue}{rgb}{0,0,0.6}  
\usepackage[colorlinks,citecolor=Myblue,linkcolor=Myblue,urlcolor=Myblue,pdfpagemode=None]{hyperref}
\usepackage{subcaption}
\usepackage{graphicx}
\usepackage{epstopdf}
\usepackage{longtable}
\usepackage[totalwidth=420pt, totalheight=590pt]{geometry}
\usepackage[inline,shortlabels]{enumitem}
\usepackage[nottoc]{tocbibind}
\usepackage{hyperref}
\usepackage{fancyhdr}
\usepackage{xcolor}
\usepackage{tikz}
\usetikzlibrary{decorations.markings}
\usetikzlibrary{knots}
\usetikzlibrary{arrows}
\usetikzlibrary{cd}

\allowdisplaybreaks

\theoremstyle{definition}
\newtheorem{defn}{Definition}
\newtheorem{thm}[defn]{Theorem}
\newtheorem{prp}[defn]{Proposition}
\newtheorem{lem}[defn]{Lemma}

\newtheorem{rem}[defn]{Remark}
\newtheorem{example}[defn]{Example}

\numberwithin{equation}{section}
\numberwithin{defn}{section}
\numberwithin{figure}{section}

\newcommand{\pic}[2][0.75]{
	\begin{tikzpicture}[scale=0.5,baseline={([yshift=-.5ex]current bounding box.center)}]
	\node at (0,0) {\includegraphics[scale=#1]{figures/#2}};
	\end{tikzpicture}
}

\begin{document}
\def\it{\textit}
\def\mcA{\mathcal{A}}
\def\mcB{\mathcal{B}}
\def\mcC{\mathcal{C}}
\def\mcD{\mathcal{D}}
\def\mcE{\mathcal{E}}
\def\mcF{\mathcal{F}}
\def\mcI{\mathcal{I}}
\def\mcL{\mathcal{L}}
\def\mcM{\mathcal{M}}
\def\mcN{\mathcal{N}}
\def\mcS{\mathcal{S}}
\def\mcV{\mathcal{V}}
\def\mcW{\mathcal{W}}
\def\mcZ{\mathcal{Z}}
\def\mcACA{{_A \mathcal{C}_A}}
\def\mcACB{{_A \mathcal{C}_B}}
\def\mcD{\mathcal{D}}
\def\mbs{\mathbf{s}}
\def\mbsp{\mathbf{s}\mathbf{'}}
\def\mbS{\mathbf{S}}
\def\mcs{\mathcal{s}}
\def\mbcs{\mathbcal{s}}
\def\mbcsp{\mathbcal{s}\mathbf{'}}
\def\mbcspp{\mathbcal{s}\mathbf{''}}
\def\euT{\mathscr{T}}
\def\opk{\Bbbk}
\def\opA{\mathbb{A}}
\def\opC{\mathbb{C}}
\def\opD{\mathbb{D}}
\def\opR{\mathbb{R}}
\def\opZ{\mathbb{Z}}
\def\opid{\mathbbm{1}}
\def\a{\alpha}
\def\b{\beta}
\def\g{\gamma}
\def\d{\delta}
\def\D{\Delta}
\def\vareps{\varepsilon}
\def\l{\lambda}
\def\abar{\overline{\a}}
\def\bbar{\overline{\b}}
\def\gbar{\overline{\g}}
\def\ddbar{\overline{\d}}
\def\mubar{\overline{\mu}}
\def\pibar{{\bar{\pi}}}

\def\opp{{\operatorname{op}}}
\def\id{\operatorname{id}}
\def\im{\operatorname{im}}
\def\Hom{\operatorname{Hom}}
\def\End{\operatorname{End}}
\def\tr{\operatorname{tr}}
\def\ev{\operatorname{ev}}
\def\coev{\operatorname{coev}}
\def\evt{\widetilde{\operatorname{ev}}}
\def\coevt{\widetilde{\operatorname{coev}}}
\def\Id{\operatorname{Id}}
\def\Vect{\operatorname{Vect}}
\def\loc{{\operatorname{loc}}}
\def\orb{{\operatorname{orb}}}
\def\coker{{\operatorname{coker}}}
\def\Ind{{\operatorname{Ind}}}
\def\Irr{{\operatorname{Irr}}}
\def\Dim{\operatorname{Dim}}
\def\pd{\partial}
\def\Bord{{\operatorname{Bord}}}
\def\sk{{\operatorname{sk}}}
\def\rd{{\operatorname{rd}}}
\def\rib{{\operatorname{rib}}}
\def\Int{\operatorname{Int}}
\def\gen{{\operatorname{gen}}}
\def\Str{\operatorname{Str}}
\def\Bordrib{\operatorname{Bord}^{\operatorname{rib}}}
\def\Bordstrat{\operatorname{Bord}^{\operatorname{strat}}}
\def\Borddef{\operatorname{Bord}^{\operatorname{def}}}
\def\Bordadm{\operatorname{Bord}^{\operatorname{adm}}}
\def\Sigmain{\Sigma_{\operatorname{in}}}
\def\Sigmaout{\Sigma_{\operatorname{out}}}
\def\rev{{\operatorname{rev}}}
\def\dol{{\operatorname{dol}}}
\def\ab{{\operatorname{ab}}}
\def\Quad{{\operatorname{Quad}}}
\def\FPdim{\operatorname{FPdim}}
\def\Gr{\operatorname{Gr}}
\def\Alg{\operatorname{Alg}}
\def\FrobAlg{\operatorname{FrobAlg}}

\def\Lra{\Leftrightarrow}
\def\Ra{\Rightarrow}
\def\ra{\rightarrow}
\def\la{\leftarrow}
\def\lra{\leftrightarrow}
\def\xra{\xrightarrow}

\newcommand{\ZRT}[1]{
	Z^{{\operatorname{RT}}}_{#1}
}
\newcommand{\Zdef}[1]{
	Z^{{\operatorname{def}}}_{#1}
}
\newcommand{\Zorb}[2]{
	Z^{{\operatorname{orb}#2}}_{#1}
}
\newcommand{\Bordribhat}[1]{
	\widehat{\operatorname{Bord}^{{\operatorname{rib}}}_{#1}}
}
\newcommand{\Borddefhat}[1]{
	\widehat{\operatorname{Bord}^{{\operatorname{def}}}_{#1}}
}

\title{\begin{flushright}
		\vspace{-1.8cm} \normalfont{\small{\textsf{ZMP-HH/21-10}}}\\
		\vspace{-0.5cm} \normalfont{\small{\textsf{Hamburger Beiträge zur Mathematik Nr.\! 898}}}\\
		\vspace{-0.5cm} \normalfont{\small{\textsf{May 2021}}}
	\end{flushright}
	\vspace{0.5cm}
	Domain walls between 3d phases of Reshetikhin-Turaev TQFTs}

\author{
	Vincent Koppen$^*$ \quad
	Vincentas Mulevi\v{c}ius$^\dagger$\\
	Ingo Runkel$^\dagger$ \quad
	Christoph Schweigert$^\dagger$\\[0.5cm]
	\normalsize{\texttt{\href{mailto:vincent.koppen@posteo.de}{vincent.koppen@posteo.de}}} \quad
	\normalsize{\texttt{\href{mailto:vincentas.mulevicius@uni-hamburg.de}{vincentas.mulevicius@uni-hamburg.de}}} \\
	\normalsize{\texttt{\href{mailto:ingo.runkel@uni-hamburg.de}{ingo.runkel@uni-hamburg.de}}} \quad
	\normalsize{\texttt{\href{mailto:christoph.schweigert@uni-hamburg.de}{christoph.schweigert@uni-hamburg.de}}}\\[0.1cm]
	{\normalsize\slshape $^\dagger$Fachbereich Mathematik, Universit\"{a}t Hamburg, Germany}\\[-0.1cm]
	{\normalsize\slshape $^*$School of Mathematics, University of Leeds, United Kingdom}\\[-0.1cm]
}

\date{}
\maketitle

\begin{abstract}
We study surface defects in three-dimensional topological quantum field theories which separate different theories of Reshetikhin-Turaev type.
Based on the new notion of a Frobenius algebra over two commutative Frobenius algebras, we present an explicit and computable construction of such defects.
It specialises to the construction in \cite{CRS2} if all $3$-strata are labelled by the same topological field theory. 
We compare the results to the model-independent analysis in \cite{FSV} and find agreement.
\end{abstract}
\newpage

\setcounter{tocdepth}{2}
\tableofcontents

\section{Introduction}

Defects, in particular topological defects, provide important structural insights into
quantum field theories, most notably into their symmetries and dualities, see e.g.\ \cite{FFRS2,Gaiotto:2014kfa}.
Defects in topological quantum field theories (TQFTs) are particularly amenable to a mathematically precise treatment.
In the present paper, we study such surface defects in three-dimensional TQFTs obtained by the Reshetikhin-Turaev construction, a mathematically rigorous procedure which to any modular fusion category associates a TQFT.
Surface defects in Reshetikhin-Turaev theories have numerous applications, ranging from the construction of quantum codes to the explanation of structures in representation theory, see e.g.\ \cite{Barkeshli:2012pr,Fuchs:2019our}.
It is very natural in this context not to restrict oneself to a single modular fusion category, but rather to consider defects separating different modular fusion categories.
The aim of the present paper is to provide additional tools for this situation.

\medskip

There are various possible points of view on defects.
A particularly simple class of defects are surface defects with the same TQFT on both sides of the defect. Among such defects, there is the transparent defect.
We make the additional simplifying assumption that a given surface defect has the property that on a small disc in the defect surface, it can be replaced by the transparent defect, i.e. one can ``punch holes''.
Repeating this process, one can replace the surface defect by the remaining network of ribbons (see Figure~\ref{fig:ribbonisation}).
Label such a ribbon by a symmetric $\D$-separable Frobenius algebra.
Then the computation of an invariant with a surface defect can be reduced to the computation of an invariant of a manifold containing a ribbon graph labelled by a Frobenius algebra.
This idea was proposed in \cite{KS} to relate the triangulations appearing in the TQFT description of correlators of rational 2d conformal field theory \cite{FRS1} to surface defects in three-dimensional TQFTs.

Mathematically, defects can be described by a \textit{defect TQFT}, which is an axiomatisation analogous to that of Atiyah-Segal, but applied to stratified manifolds \cite{CRS1}.
The above mentioned construction of surface defects in TQFTs of Reshetikhin-Turaev type, as well as its generalisation to include line and point defects, was turned into a defect TQFT in \cite{CRS2}.
\begin{figure}
\captionsetup{format=plain, indention=0.5cm}
\centering
\pic[1.5]{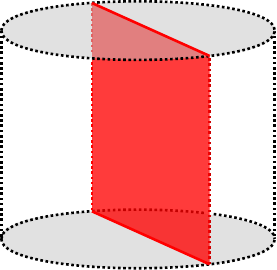} \hspace{-0.25cm}$\rightsquigarrow$\hspace{-0.25cm}
\pic[1.5]{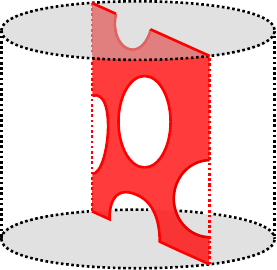} \hspace{-0.25cm}$\rightsquigarrow$\hspace{-0.25cm}
\pic[1.5]{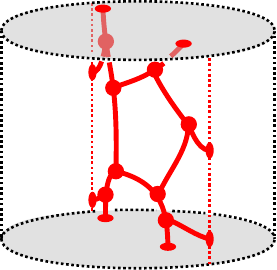}
\caption{
Converting a surface defect into a ribbon graph by ``punching holes''.
The resulting graph can be labelled in a natural way by a symmetric $\D$-separable Frobenius algebra (see \cite[Sec.\,6.2]{FSV}).
}
\label{fig:ribbonisation}
\end{figure}

\medskip

A different point of view allows one to also address surface defects that separate {\em non-identical} topological field theories, so-called \textit{domain walls}.
In this situation, no transparent surface defect exists, and hence
no holes can be punched.
A model-independent analysis of such defects, under the hypothesis that a defect TQFT exists, was carried out in \cite{FSV}.
It leads to the insight that the existence of surface defects is, in general, obstructed, with an obstruction in the Witt group of modular fusion categories.
The analysis of \cite{FSV} also yields a very natural candidate for the bicategory of surface defects.

\medskip

In Section~\ref{sec:2} we give an explicit construction of a defect TQFT including domain walls between 3d regions labelled by different modular categories. 
All modular fusion categories involved have to be in the same Witt class.
To describe two modular fusion categories in the same Witt class, we fix one such category $\mcC$ and choose two symmetric $\D$-separable commutative haploid Frobenius algebras $A$, $B$.
It is shown in \cite{DMNO} that the categories of local modules $\mcC^\loc_A$ and $\mcC^\loc_B$ are modular fusion categories, both in the same Witt class as $\mcC$.
Using again the idea of ``punching holes'', but now in the bulk, we obtain a network of line and surface defects -- a foam -- in the TQFT for the modular category $\mcC$.
These bulk foams are labelled by the algebras $A$ and $B$, respectively.
\begin{figure}
\captionsetup{format=plain, indention=0.5cm}
\centering
\pic[1.5]{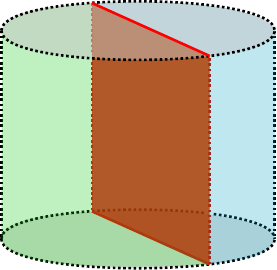} \hspace{-0.25cm}$\rightsquigarrow$\hspace{-0.25cm}
\pic[1.5]{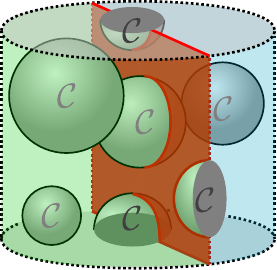} \hspace{-0.25cm}$\rightsquigarrow$\hspace{-0.25cm}
\pic[1.5]{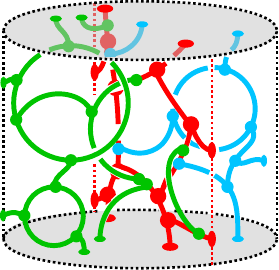}
\caption{
The procedure of ``punching holes'' in the bulk theories and domain walls yields a foam of defects.
Upon evaluating it with the TQFT obtained from a modular fusion category $\mcC$, the foam is further reduced to a ribbon graph.
}
\label{fig:foamification}
\end{figure}

The domain wall that separates these bulk theories can now also be resolved into a network, since the modular fusion category $\mcC$ provides a common ground (see Figure~\ref{fig:foamification}).
To label this network, we introduce a new algebraic notion in Definition~\ref{def:Frob_algs_over_AB}: given two symmetric $\D$-separable commutative haploid
Frobenius algebras in a modular fusion category $\mcC$, a (\it{symmetric, $\D$-separable})
\it{Frobenius algebra over $(A,B)$} is a (symmetric, $\D$-separable) Frobenius algebra
$F\in\mcC$, which is also an $A$-$B$-bimodule, such that product and braiding are
compatible.
(Left and right action are treated with the braiding and its inverse in this
definition.)
Such a Frobenius algebra over $(A,B)$ turns out to be the correct label for a defect separating the bulk theories labelled by the modular fusion categories $\mcC^\loc_A$ and $\mcC^\loc_B$.
Furthermore, in Definition~\ref{def:multimodule} we introduce the notion of a multimodule over suitable $n$-tuples of Frobenius algebras over commutative algebras.
They are the labels for generalised Wilson lines at which several domain walls meet.

One useful feature of our construction is that one can consider configurations with three or more bulk theories.
In fact any configuration of bulk theories can be evaluated if one chooses the common category $\mcC$ large enough.
The need to adapt the ambient category $\mcC$ depending on the configuration of bulk theories can be eliminated by considering generalised orbifold data \cite{CRS3} (of which the commutative algebras are a special case).
However, the resulting construction would not include situations that cannot be described by the present construction as well.

In the subsequent Sections~\ref{sec:3} and \ref{sec:4} we compare the outcome of this construction to the model independent analysis in \cite{FSV}.
In particular, we write down an explicit Witt trivialisation $\mcC^\loc_A \boxtimes \widetilde{\mcC^\loc_B}\simeq\mcZ(\mcACB)$.
It implies that module categories of $\mcACB$ are classified by algebras over $(A,B)$, and it shows that from the TQFT constructed in Section~\ref{sec:2} one obtains the bicategory of surface defects found in \cite{FSV}.
This is the content of Theorem~\ref{thm:bicat-equiv}.

\medskip

\subsubsection*{Acknowledgements}

We thank Nils Carqueville and Gregor Schaumann for helpful comments on a draft of this paper.
All authors have been partially supported by the RTG 1670 ``Mathematics inspired by String theory and Quantum Field Theory''.
VM, IR, CS are partially supported by the Deutsche Forschungsgemeinschaft (DFG, German Research Foundation) under Germany's Excellence Strategy - EXC 2121 ``Quantum Universe'' - 390833306.

\section{Domain walls via algebras}
\label{sec:2}

\subsection{Multifusion categories}
\label{subsec:fusion_cats}
Throughout this paper we will use the language of multifusion categories extensively, see e.g.\ \cite[Sec.\,4]{EGNO}.
A \textit{multifusion category} $\mcA$ over the field $\opC$ of complex numbers 
(or any algebraically closed field of characteristic zero)
is 
\begin{itemize}
\item finitely semisimple:
the sets of morphisms $\mcA(X,Y) = \Hom_\mcA(X,Y)$
are finite dimensional complex vector spaces for all $X,Y\in\mcC$,
any finite collection of objects has a direct sum,
there are finitely many (isomorphism classes of) simple objects,
every object is isomorphic to a direct sum of finitely many simple objects (the empty direct sum is the zero object $0\in\mcA$);
\item monoidal:
$\mcA$ is equipped with a tensor product $\otimes:\mcA\times\mcA\ra\mcA$ which has a unit $\opid=\opid_\mcA \in \mcA$
and the associator and unitor natural isomorphisms satisfying the usual pentagon and triangle identities; the tensor product of morphisms is assumed to be bilinear;
\item rigid:
each object $X\in\mcA$ has a distinguished left and right dual, i.e.\ objects $X^*, {}^*\!X \in \mcA$ together with evaluation morphisms $\ev_X\colon X^* \otimes X \ra \opid$, $\evt_X\colon X \otimes {}^*\!X \ra \opid$ and coevaluation morphisms $\coev_X\colon \opid \ra X \otimes X^*$ and $\coevt_X\colon {}^*\!X \otimes X \ra \opid$.
\end{itemize}
The tensor unit $\opid_\mcA$ of a multifusion category need not be a simple object. If $\opid_\mcA$ is simple, i.e.\ if $\dim\End_\mcA(\opid_\mcA) = 1$, $\mcA$ is called a \textit{fusion category}.

\medskip

Recall that for two finitely semisimple categories $\mcA$, $\mcB$ their Deligne product $\mcA\boxtimes\mcB$ is the semisimple category consisting of formal direct sums of objects of the form $X \boxtimes Y$, $X\in\mcA$, $Y\in\mcB$, such that direct sums in $\mcA$ and $\mcB$ distribute with respect to the symbol $\boxtimes$.
Morphism spaces in $\mcA\boxtimes\mcB$ are tensor products of vector spaces of morphisms in $\mcA$ and $\mcB$.
If $\mcA$ and $\mcB$ are multifusion so is $\mcA\boxtimes\mcB$ with the tensor unit $\opid_\mcA \boxtimes \opid_\mcB$, see \cite[Cor.\,4.6.2]{EGNO}.
The direct sum $\mcA\oplus\mcB$ is the finitely semisimple category of pairs $(X,Y)$ with the direct sum defined additively on each entry.
It is customary to identify e.g.\ $\mcA$ with the collection of objects $(X,0)\in\mcA\oplus\mcB$, $X\in\mcA$.
For $\mcA$, $\mcB$ multifusion, $\mcA\oplus\mcB$ is again multifusion with the tensor unit $\opid_\mcA \oplus \opid_\mcB$.
A multifusion category which cannot be written as a direct sum of two multifusion categories in a non-trivial way is called \textit{indecomposable}
\cite[Sec.\,2.4]{ENO}.
We note that a multifusion category can have a non-simple tensor unit and still be indecomposable. The finite-dimensional bimodules of the semisimple $\opC$-algebra $\opC \oplus \opC$ are an example of this.

\medskip

The simple summands of the tensor unit of a multifusion category can be used to decompose it (see \cite[Rem.\,4.3.4]{EGNO}):
\begin{prp}
\label{prp:multifusion_cats}
Let $\mcA$ be a multifusion category.
\begin{enumerate}[i)]
\item \label{prp:multifusion_cats:i}
Let $\opid_\mcA = \bigoplus_i \opid_i$ be a decomposition of the tensor unit into simple objects.
Then $\opid_i$ are mutually non-isomorphic with $\opid_i \otimes \opid_j \cong \d_{ij}\opid_i$.
\item \label{prp:multifusion_cats:ii}
There is a decomposition of linear categories $\mcA = \bigoplus_{ij} \mcA_{ij}$, where $\mcA_{ij} := \opid_i \otimes \mcA \otimes \opid_j$.
Each $\mcA_{ii}$ is a fusion category.
\end{enumerate}
\end{prp}

\medskip

A functor $F\colon\mcA\ra\mcB$ between multifusion categories $\mcA$ and $\mcB$ is assumed to be linear on morphism spaces and have a monoidal structure, i.e.\ a natural isomorphism $F(-) \otimes F(-) \Ra F(-\otimes -)$ and an isomorphism $F(\opid_\mcA) \cong \opid_\mcB$ satisfying the usual compatibility conditions.

\medskip

For a multifusion category $\mcA$ we will denote by $\Irr_\mcA$ a set of representatives of isomorphism classes of its simple objects.
If $\mcA$ is fusion, we in addition assume that $\opid_\mcA\in\Irr_\mcA$.
For $i\in\Irr_\mcA$, we will use the identification $\End_\mcA i \cong \opC$, $\id_i \mapsto 1$ when necessary.
Recall that the Grothendieck ring $\Gr(\mcA)$ of $\mcA$ is the commutative ring with $\opZ$-basis given by the set $\Irr_\mcA$ 
together with the product
\begin{equation}
i \cdot j = \sum_{k\in\Irr_\mcA} N_{ij}^k ~k, \quad
i,j\in\Irr_\mcA ~,~
N_{ij}^k := \dim_\opC \mcA(i\otimes j, k) ~.
\end{equation}
If $\mcA$ is fusion, the Frobenius-Perron dimension is defined as the unique ring homomorphism $\FPdim\colon\Gr(\mcA)\ra\opR$, such that $\FPdim(i) > 0$ for all $i\in\Irr_\mcA$.
The Frobenius-Perron dimension of $\mcA$ is defined as $\FPdim(\mcA) := \sum_{i\in\Irr_\mcA} (\FPdim(i))^2$.
The following result then provides a convenient way to prove equivalence of two fusion categories, see \cite[Prop.\,6.3.3]{EGNO}.

\begin{prp}
\label{prp:equiv_ito_FPdims}
A functor $F\colon\mcA\ra\mcB$ between fusion categories $\mcA$ and $\mcB$ is an equivalence iff it is fully faithful (i.e.\ isomorphism on morphism spaces) and $\FPdim(\mcA) = \FPdim(\mcB)$.
\end{prp}

\medskip

\begin{figure}
\captionsetup{format=plain, indention=0.5cm}
\centering
\begin{subfigure}[b]{0.98\textwidth}
    \centering
    \pic[1.25]{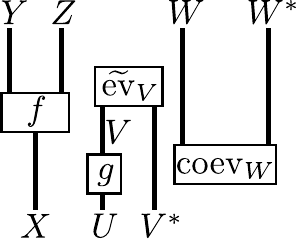} $=$
    \pic[1.25]{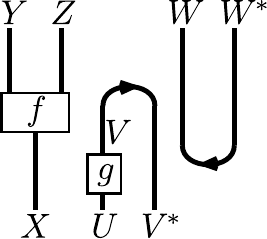} $=$
    \pic[1.25]{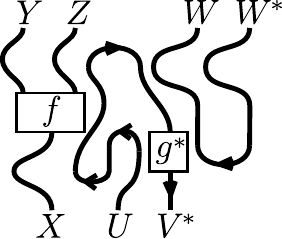}
    \caption{}
    \label{fig:graphical_calculus_pivotal}
\end{subfigure}
\begin{subfigure}[b]{0.98\textwidth}
    \centering
    \pic[1.25]{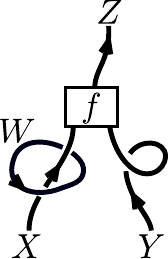} =
    \pic[1.25]{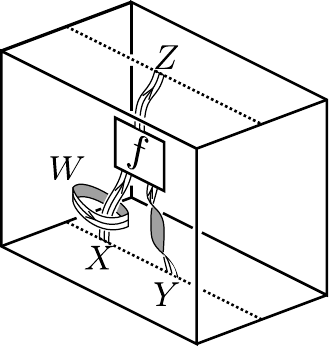}
    \caption{}
    \label{fig:graphical_calculus_ribbon}
\end{subfigure}\\
\caption{
(a) Graphical calculus for a pivotal category $\mcA$.
A morphism is depicted by a string diagram, consisting of \textit{strands} labelled by objects and \textit{coupons} labelled by morphisms.
A diagram is read from bottom to top, in this example the diagram depicts a morphism
$(\id_{Y\otimes Z} \otimes \evt_V \otimes \id_{W\otimes W^*}) \circ (f\otimes g \otimes \id_{V^*} \otimes \coev_W)$
where $X,Y,Z,U,V,W\in\mcA$, $f\colon X\ra Y\otimes Z$, etc.
It is customary (although not required) to omit strands labelled with the tensor unit and coupons labelled with associators, unitors and identity objects, and to replace (co)evaluation morphisms with bent lines.
In the latter case one also adds directions to strands; a downwards direction corresponds to the dual of an object.
The axioms of a pivotal category imply that string diagrams up to a plane isotopy  with fixed  ends of incoming and outgoing strands yield equal morphisms.\\
\vspace{-0.25cm}\\
(b) 
Graphical calculus for a ribbon category.
In this case the strands can be seen as directed and \textit{framed} lines with framing given by the paper plane.
The diagrams can be deformed as if embedded in $(0,1) \times (0,1) \times [0,1]$ with incoming strands starting at fixed points on the line $\{1/2\}\times(0,1)\times\{0\}$ and the outgoing ones ending at fixed points on $\{ 1/2 \}\times(0,1)\times\{1 \}$.
The framing is encoded by replacing the strands by ribbons.
}
\label{fig:graphical_calculus}
\end{figure}
We say that a multifusion category $\mcA$ is \textit{pivotal} if it is endowed with a \textit{pivotal structure}, i.e.\ a monoidal natural isomorphism $\d\colon\Id_\mcA \Ra (-)^{**}$.
Since in any case $X \cong ({}^*\!X)^*$ for all objects $X$ (see \cite[Rem.\,2.10.3]{EGNO}), this implies $X^* \cong ({}^*\!X)^{**} \cong {}^*\!X$. It is therefore enough to only consider one of the two duals in pivotal categories, and we will use $X^*$.
An equivalent way to define a pivotal structure is to fix for each object $X$ a (left and right) dual $X^*$ such that the identities
\begin{equation}
\label{eq:pivotal_conds}
\pic[1.25]{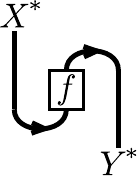} = \pic[1.25]{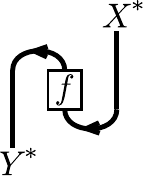} (=: f^*)\, ,
\pic[1.25]{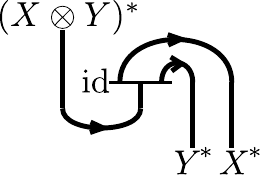} = \pic[1.25]{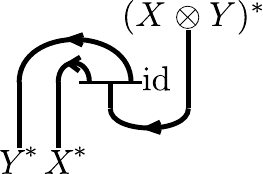}
\end{equation}
hold (see e.g.\ \cite[Lem.\,2.12]{CR}).
Here we use graphical calculus to depict morphisms by string diagrams, our conventions are summarised in the example in Figure~\ref{fig:graphical_calculus_pivotal}.

\medskip

To any object $X\in\mcA$ and a morphism $f\in\End_\mcA X$ we define the left and the right traces to be
\begin{equation}
\label{eq:left_right_traces}
\tr_l f = \pic[1.25]{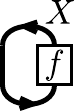} ~, \quad
\tr_r f = \pic[1.25]{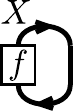} ~.
\end{equation}
One also defines the left and right dimensions of an object $X\in\mcA$ as $\dim_l(X) := \tr_l \id_X$ and $\dim_r(X) := \tr_r \id_X$. 
A pivotal category $\mcA$ is called spherical \cite{BW} if $\tr_l f = \tr_r f (=: \tr f)$ for all $f\in\End_\mcA X$.
In this case we define the dimension of an object $X\in\mcA$ as $\dim X := \tr \id_X$.
The sphericality condition can be simplified as follows, see e.g.\ \cite[Lem.\,4.4]{TuVi}.\footnote{
In this reference the fusion case is considered, but the argument for the multifusion case is the same.
}
\begin{prp}
\label{prp:spherical_ito_simples}
A pivotal multifusion category $\mcA$ is spherical iff $\dim_l i = \dim_r i$ for all $i\in\Irr_\mcA$.
\end{prp}

A braiding on a multifusion category $\mcA$ is a natural isomorphism $\{c_{X,Y}\colon X\otimes Y \ra Y \otimes X\}_{X,Y\in\mcA}$ satisfying the two hexagon identities.
In the graphical calculus the braiding will be denoted by
\begin{equation}
c_{X,Y} = \pic[1.25]{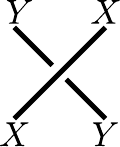} ~, \quad
c^{-1}_{X,Y} = \pic[1.25]{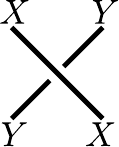} ~.
\end{equation}
For a braided multifusion category $\mcA$ we will denote by $\widetilde{\mcA}$ the reverse of it, i.e.\ $\mcA$ equipped with the reverse braiding $\{c^{-1}_{Y,X}\colon X\otimes Y \ra Y\otimes X\}_{X,Y\in\mcA}$.
A braided functor between braided categories is a monoidal functor which preserves the braiding morphisms.

\medskip

For a braided pivotal multifusion category one defines the left and right twist morphisms and their inverses as
\begin{equation}
\theta^l_X = \pic[1.25]{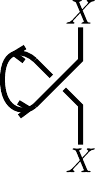} ~,~
\theta^r_X = \pic[1.25]{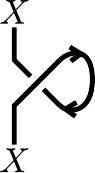} ~, \quad
(\theta^l_X)^{-1} = \pic[1.25]{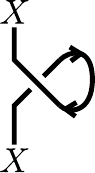} ~,~ 
(\theta^r_X)^{-1} = \pic[1.25]{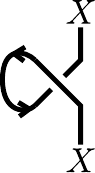} ~. 
\end{equation}
A braided pivotal multifusion category $\mcA$ is called ribbon if $\theta_X^l = \theta_X^r (=: \theta_X)$ for all $X\in\mcA$.
A ribbon functor between ribbon categories is a braided functor which preserves twists.
Graphical calculus allows one to represent a morphism in a ribbon category by a ribbon tangle embedded in $(0,1)\times (0,1)\times [0,1]$, see Figure~\ref{fig:graphical_calculus_ribbon} for an example.
For more details on braiding and ribbon structures see \cite[Sec.\,3.3]{TuVi}
The following result is a convenient criterion when checking the ribbon condition:
\begin{prp}
\label{prp:ribb_iff_spherical}
A braided pivotal multifusion category $\mcA$ is ribbon iff it is spherical.
\end{prp}
\begin{proof}
In the fusion case, this is shown e.g.\ in \cite[Lem.\,4.5]{TuVi}. In the multifusion case, one can decompose $\mcA$ as in Proposition~\ref{prp:multifusion_cats}, and since $\mcA$ is braided, $\mcA_{ij} = 0$ for $i \neq j$. Thus $\mcA = \bigoplus_i \mcA_{ii}$ is a sum of fusion categories and one can apply the previous argument in each summand.
\end{proof}

\medskip

Given a multifusion category $\mcA$, its Drinfeld centre $\mcZ(\mcA)$ is the category of pairs $(X,\g)$, where $X\in\mcA$ and $\g\colon X\otimes - \Ra -\otimes X$ is a half-braiding, i.e.\ a natural isomorphism satisfying the hexagon identity.
Morphisms in $\mcZ(\mcA)$ are morphisms in $\mcA$ commuting with half-braidings.
$\mcZ(\mcA)$ is multifusion\footnote{
For fusion categories this was shown in \cite[Thm.\,3.16]{Mu}.
For multifusion categories $\mcA$ one can use that $\mcZ(\mcA)$ is equivalent to the category of bimodule-endofunctors of $\mcA$ seen as an $\mcA$-$\mcA$-bimodule over itself, and that the latter category is semisimple by \cite[Thm.\,2.18]{ENO}.
} 
and braided with braiding $c_{(X,\g),(Y,\d)} = \g_Y$.
If $\mcA$ is pivotal/spherical, then there is a natural pivotal/spherical structure on $\mcZ(\mcA)$.
In particular, if $\mcA$ is spherical, then by Proposition~\ref{prp:ribb_iff_spherical}, $\mcZ(\mcA)$ is also ribbon.
The converse is not true: 
Below we will encounter examples 
where $\mcA$ is not spherical but $\mcZ(\mcA)$ is, cf.\ Remark~\ref{rem:ACB_not_spherical}.
If $\mcA$ is fusion, the Frobenius-Perron dimension of the Drinfeld centre is known to be $\FPdim(\mcZ(\mcA)) = \FPdim(\mcA)^2$.
Using the notation of Proposition~\ref{prp:multifusion_cats} we have (see e.g.\  \cite[Thm.\,2.5.1]{KZ}):
\begin{prp}
\label{prp:multifusion_center}
If $\mcA$ is an indecomposable multifusion category, then $\mcZ(\mcA)$ is a braided fusion category and is braided equivalent to $\mcZ(\mcA_{ii})$ for each $i$.
\end{prp}

\medskip

Finally, for applications to TQFTs the notion of a \textit{modular fusion category (MFC)} is important.
A MFC is a ribbon fusion category whose braiding is \textit{non-degenerate}, i.e.\ if for some $T\in\mcC$ and for all $X\in\mcC$ the identity
\begin{equation}
c_{T,X} \circ c_{X,T} = \id_{X\otimes T}
\end{equation}
holds, then already $T \cong \opid^{\oplus n}$ for some $n\ge 0$.
In other words, all transparent objects are isomorphic to direct sums of the tensor unit.\footnote{
The original definition of modularity is via non-degeneracy of the 
$s$-matrix, i.e.\ the matrix with entries $\tr(c_{j,i}\,c_{i,j})$, $i,j \in \Irr_\mcC$, see \cite[Sec.\,II.1.4]{Tu}. The relation to transparent objects can be found in \cite{Br}.
}
The following statement is convenient when looking for equivalences between non-degenerate (and therefore also modular) braided fusion categories (see \cite[Cor.\,3.26]{DMNO}):
\begin{prp}
\label{prp:non-deg_fully-faith}
A braided functor $F\colon\mcC\ra\mcD$ between braided fusion categories $\mcC$, $\mcD$ with $\mcC$ non-degenerate is automatically fully-faithful.
\end{prp}

An equivalent condition for $\mcC$ to be a MFC is that the braided functor
\begin{equation}
\label{eq:CC_to_ZC_eq}
\mcC \boxtimes \widetilde{\mcC} \ra \mcZ(\mcC), \qquad
X \boxtimes Y \mapsto (X \otimes Y, \g_{X,Y}^\dol) ~,
\end{equation}
is an equivalence, see \cite[Prop.\,8.20.12]{EGNO}.
Here, $\g_{X,Y}^\dol$ is the ``dolphin'' half-braiding defined for all $U\in\mcC$ by
\begin{equation}
\label{eq:dolphin_half-br}
(\g_{X,Y}^\dol)_U := \pic[1.25]{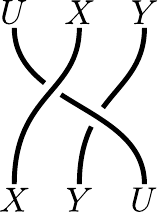}  \, .
\end{equation}
The Drinfeld centre $\mcZ(\mcS)$ of a spherical fusion category $\mcS$ is a MFC (\cite{Mu}, see also e.g.\ \cite[Cor.\,8.20.14]{EGNO}).

\subsection{Algebras and modules}
\label{subsec:algs_and_modules}
An algebra in a multifusion category $\mcC$ is an object $A\in\mcC$ together with a multiplication morphism $\mu\colon A\otimes A \ra A$ and a unit morphism $\eta\colon \opid_\mcC \ra A$, such that
\begin{equation}
\pic[1.25]{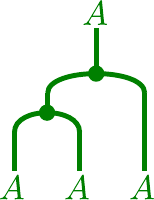} = \pic[1.25]{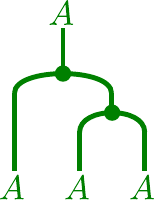}, \quad
\pic[1.25]{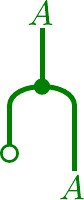} = \pic[1.25]{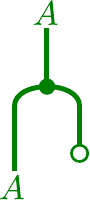} = \id_A ~.
\end{equation}
We will call $A$ a haploid algebra if $\dim \mcC(\opid_\mcC,A) = 1$.
If $\mcC$ is braided, one says that $A$ is a commutative algebra if
\begin{equation}
\pic[1.25]{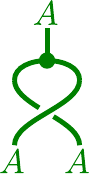} = \pic[1.25]{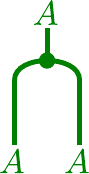}
\end{equation}
Similarly one defines a coalgebra as an object with comultiplication and counit morphisms satisfying coassociativity and counitality.

\medskip

In what follows we will make use of the notion of a Frobenius algebra extensively.
It is an object $A$, which is simultaneously an algebra and a coalgebra, such that
\begin{equation}
\label{eq:Frob_property}
\pic[1.25]{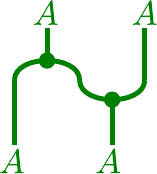} = \pic[1.25]{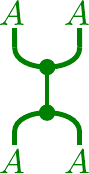} = \pic[1.25]{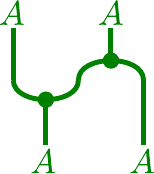} ~.
\end{equation}
We denote the coproduct and counit by $\D \colon A \ra A \otimes A$ and $\varepsilon \colon A \ra \opid_\mcC$ respectively.
A Frobenius algebra $A$ is called
\begin{equation}
\label{eq:sep_and_symm_properties}
\text{$\D$-separable if } \pic[1.25]{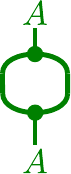} = \id_A ~,\quad
\text{symmetric if} \pic[1.25]{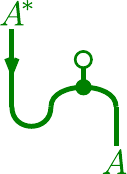} = \pic[1.25]{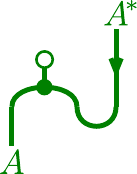} ~,
\end{equation}
where the latter notion requires $\mcC$ to be pivotal.
For the dimension of $A$ we have the following result:

\begin{lem}\label{lem:dimA-nonzero}
Let $\mcC$ be a spherical multifusion category and $A$ a haploid $\Delta$-separable symmetric Frobenius algebra in $\mcC$. Then $\dim(A) \neq 0$.
\end{lem}

\begin{proof}
Since the unit $\eta$ and counit $\vareps$ of $A$
are both non-zero morphisms and $A$ is haploid one has $\vareps\circ\eta \neq 0$.
On the other hand, for symmetric $\D$-separable Frobenius algebras one has $\vareps\circ\eta = \dim A$ 
(see e.g.\ \cite[Eq.\,(3.49)]{FRS1}).
\end{proof}

We note that one cannot drop the requirement that $A$ is symmetric.
For example, the two-dimensional Clifford algebra in super-vector spaces is haploid $\Delta$-separable Frobenius, but has quantum dimension $\dim(A)=0$.
A Frobenius algebra is called \textit{special} if $\mu \circ \D$ and $\varepsilon \circ \eta$ are non-zero multiples of the respective identity morphisms.
Thus if $\mcC$ is fusion, Frobenius algebras as in Lemma~\ref{lem:dimA-nonzero} are special.

\medskip

A (left-)module of an algebra $A$ is an object $M$ together with an action morphism $\rho_M\colon A\otimes M \ra M$, such that
\begin{equation}
\pic[1.25]{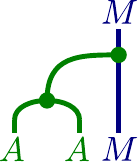} = \pic[1.25]{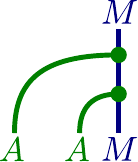} ~, \quad
\pic[1.25]{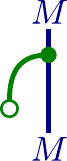} = \id_M ~.
\end{equation}
Similarly one defines right modules and bimodules.
For two modules $M$ and $N$, a module morphism is a morphism $f\colon M\ra N$ in $\mcC$ which commutes with the actions.
In case $A$ is a Frobenius algebra, every module $M$ is canonically a comodule with the coaction
\begin{equation}
\pic[1.25]{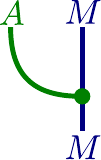} := \pic[1.25]{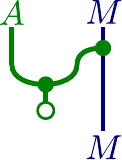} ~.
\end{equation}

\medskip

Modules and bimodules and their morphisms form categories, which we denote by $\mcC_A$ and $\mcACA$.
If $A$ is a $\D$-separable Frobenius algebra, both of them are finitely semisimple (see \cite[Prop.\,5.24]{FS}) and $\mcACA$ is multifusion with the tensor product of two bimodules $M$ and $N$ being
\begin{equation}
M \otimes_A N := \im P_{M\otimes N} ~, \text{ where }
P_{M\otimes N} = \pic[1.25]{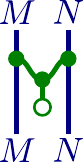} \stackrel{(*)}{=} \pic[1.25]{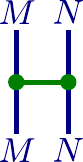} ~.
\end{equation}
In $(*)$ we have introduced a simplified notation for the projector $P_{M\otimes N}$.
The tensor unit is defined to be $\opid_{\mcACA} := A$.
If $\mcC$ is pivotal, so is $\mcACA$ with evaluation/coevaluation morphisms
\begin{equation}
\ev_M = \pic[1.25]{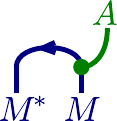}, \,
\coev_M = \pic[1.25]{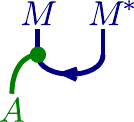}, \,
\evt_M = \pic[1.25]{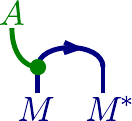}, \,
\coevt_M = \pic[1.25]{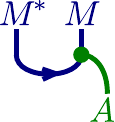}.
\end{equation}

\medskip

Let $\mcC$ be a braided fusion category and $A$ a commutative algebra in it.
Given a module $M\in\mcC_A$, one can define two right $A$-actions on it as follows
\begin{equation}
\pic[1.25]{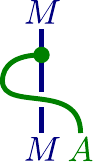}
\pic[1.25]{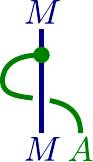}
\end{equation}
We say that $M$ is a \textit{local $A$-module} if these right actions are equal or equivalently
\begin{equation}
    \rho_M \circ c_{M,A} \circ c_{A,M} = \rho_M ~.
\end{equation}
The category of local modules will be denoted by $\mcC^\loc_A$.
If $\mcC$ is in addition a ribbon category, $\mcC^\loc_A$ is also ribbon with tensor product and dualities as in $\mcACA$ and twists of objects as in $\mcC$, see e.g.\ \cite[Sec.\,3.4]{FFRS}.
One has the following important result, see \cite{KO} and \cite[Cor.\,3.30,\,3.32]{DMNO}:
\begin{thm}
\label{thm:CAloc_is_modular}
Let $\mcC$ be a modular fusion category and $A$ a symmetric $\D$-separable commutative haploid Frobenius algebra in $\mcC$.
Then $\mcC^\loc_A$ is again a modular fusion category with the Frobenius-Perron dimension $\FPdim(\mcC^\loc_A) = \tfrac{\FPdim(\mcC)}{(\FPdim(A))^2}$.
\end{thm}

\medskip

An example illustrating the above result is the following (see \cite[Thm.\,2.3]{BK} and \cite[Lem.\,3.5]{DMNO}):

\begin{example}\label{ex:lag-alg}
For a spherical fusion category $\mcS$ there is a commutative haploid symmetric $\D$-separable Frobenius algebra $A\in\mcZ(\mcS)$ with the underlying object $A = \bigoplus_{i\in\Irr_\mcS} i \otimes i^*$.
One evidently has $\FPdim(A) = \FPdim(\mcS)$, which by Proposition~\ref{prp:equiv_ito_FPdims} and Theorem~\ref{thm:CAloc_is_modular} implies that $\mcZ(\mcS)_A^\loc \simeq \Vect$.
In general, algebras with this property are called \textit{Lagrangian}, see \cite[Sec.\,4.2]{DMNO}.
\end{example}

For the later discussion of surface defects we will need the following new notion:
\begin{defn}
\label{def:Frob_algs_over_AB}
Let $\mcC$ be a ribbon category and $A,B\in\mcC$ two symmetric $\D$-separable commutative Frobenius algebras.
A (\it{symmetric, $\D$-separable}) \it{Frobenius algebra over $(A,B)$} is a (\it{symmetric, $\D$-separable}) Frobenius algebra $F\in\mcC$, which is simultaneously an $A$-$B$-bimodule, such that 
\begin{align*}
&
\pic[1.25]{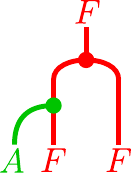} = 
\pic[1.25]{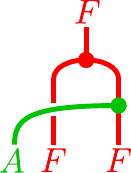} =
\pic[1.25]{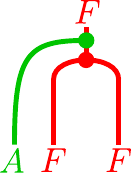}, \qquad
\pic[1.25]{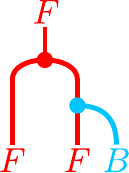} =
\pic[1.25]{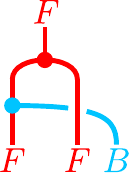} =
\pic[1.25]{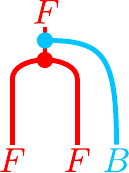},\\
&
\pic[1.25]{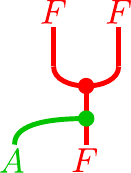} =
\pic[1.25]{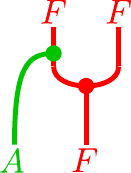} =
\pic[1.25]{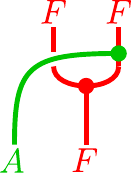}, \qquad
\pic[1.25]{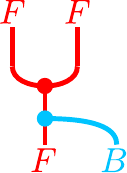} =
\pic[1.25]{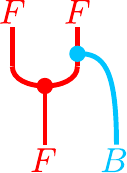} =
\pic[1.25]{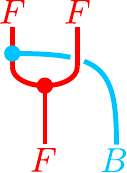}.
\end{align*}
Note that $A\otimes B$ is an example of such an algebra.
\end{defn}

\begin{rem}
\label{rem:Frob_algs_over_AB_cross_opp}
\begin{enumerate}[i), wide, labelwidth=0pt, labelindent=0pt]
\item
Because of the Frobenius property, the compatibility conditions of $A$- and $B$-actions with the product of $F$  in the above definition imply those with the coproduct, and vice versa. \item 
The reason for the choice of over- and under-crossings in the relations in Definition~\ref{def:Frob_algs_over_AB} will become apparent in Section~\ref{subsec:constr_of_defect_TQFT}.
There, a symmetric $\D$-separable Frobenius algebra $F\in\mcC$ over $(A,B)$ is used to label a surface defect separating two TQFTs labelled by $A$ and $B$.
The crossings mean that the $A$-phase is ``above'' the defect, while the $B$-phase is ``below'', assuming that the defect itself is oriented towards the $A$-phase.
\item
If $F\in\mcC$ is a Frobenius algebra over $(A,B)$ with multiplication $\mu$, comultiplication $\D$ and $A$- and $B$-actions $\l$ and $\rho$, its \textit{opposite} is a Frobenius algebra $F^{\operatorname{op}}\in\mcC$ over $(B,A)$ having the same underlying object $F$, multiplication $\mu\circ c_{F,F}$, comultiplication $c_{F,F}^{-1}\circ\D$, $B$-action $\rho\circ c_{B,F}$ and $A$-action $\l\circ c_{F,A}$, see also \cite[Sec.\,3.5]{FRS1}.
Left $F$-modules are in bijection with right $F^{\operatorname{op}}$-modules.
Opposite algebras will be used to reverse the orientation of defects.
\end{enumerate}
\end{rem}

Let $A$, $B$, $F$ be as in the Definition~\ref{def:Frob_algs_over_AB} above.
A left $F$-module $M\in\mcC$ and a right $F$-module $N\in\mcC$ are automatically $A$-$B$-bimodules with actions
\begin{align}
\label{eq:AB-action-from-F-action-1}
\pic[1.25]{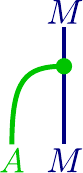} :=
\pic[1.25]{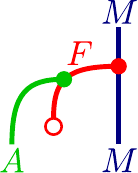}  \qquad &, \qquad
\pic[1.25]{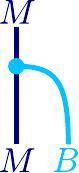} :=
\pic[1.25]{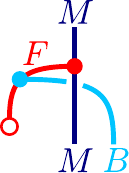} \, ,\\
\label{eq:AB-action-from-F-action-2}
\pic[1.25]{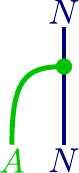} :=
\pic[1.25]{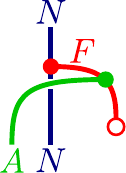}\qquad &,\qquad
\pic[1.25]{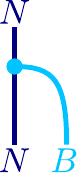} :=
\pic[1.25]{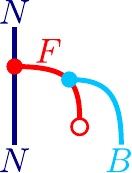} \, .
\end{align}
It is easy to check that these are the unique $A$- and $B$- actions on $M$ and $N$ such that the following identities hold:
\begin{align}
&
\pic[1.25]{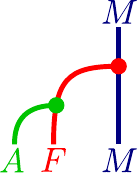} =
\pic[1.25]{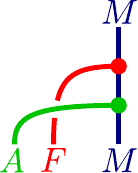} =
\pic[1.25]{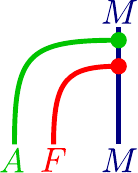}, \qquad
\pic[1.25]{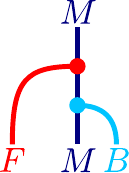} =
\pic[1.25]{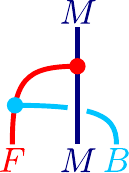} =
\pic[1.25]{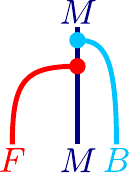},
\nonumber
\\
&
\pic[1.25]{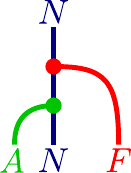} =
\pic[1.25]{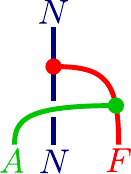} =
\pic[1.25]{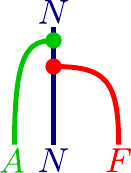}, \qquad
\pic[1.25]{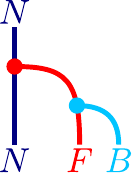} =
\pic[1.25]{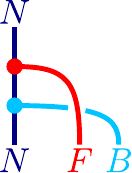} =
\pic[1.25]{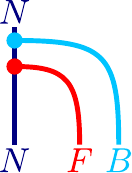}.
\label{eq:AB-module-identities}
\end{align}
A morphism of (left-, right-) modules $M$, $N$ is a morphism $f\colon M\ra N$, commuting with the $F$- (and therefore also $A$-, $B$-) actions.
Because of the above argument, there is no need to introduce ``left or right modules over $(A,B)$''. However, the situation is different for bimodules, where one obtains two a priori different $A$-$B$-bimodule structures. Since the application to line defects requires all of the identities in \eqref{eq:AB-module-identities} to hold simultaneously for a bimodule $M=N$, we define:
\begin{defn}
\label{def:bimodules_over_AB}
Let $F_1, F_2\in\mcC$ be algebras over $(A,B)$. An \textit{$F_1$-$F_2$-bimodule over $(A,B)$} is an $F_1$-$F_2$-bimodule $M$ such that the left $A$-actions and right $B$-actions induced  via \eqref{eq:AB-action-from-F-action-1} and \eqref{eq:AB-action-from-F-action-2} by $F_1$ and $F_2$ respectively, are equal.
\end{defn}

\subsection{Construction of the defect TQFT}
\label{subsec:constr_of_defect_TQFT}
A modular fusion category $\mcC$ yields a $3$-dimensional TQFT via the Reshetikhin-Turaev (RT) construction \cite{RT2,Tu}.
In it one interprets the objects of $\mcC$ as data assigned to bulk Wilson line operators (by analogy with Chern-Simons theory \cite{Wi}, where Wilson lines are labelled by representations of the gauge group).
In the $3$-sphere $S^3$, the vacuum expectation value of a link $W$ of Wilson lines is 
defined to be the number obtained by projecting $W$ on a plane and using graphical calculus to read it as a morphism in $\End_\mcC(\opid_\mcC) \cong \opC$.
The Wilson lines are framed and are represented as ribbons.
More generally one can evaluate ribbon graphs involving also coupons embedded in $S^3$.
Other $3$-manifolds can be achieved by representing the corresponding $3$-manifold as a framed link $L$ in $S^3$ using surgery and treating the components of $L$ also as Wilson lines, but labelled in a specific way.

\medskip

Properties of surface defects between Reshetikhin-Turaev models were explored in  \cite{KS,FSV} and it was understood in \cite{KS,FSV,CRS2} that a surface defect $S$ can be converted into a ribbon graph whose strands are labelled with a symmetric $\D$-separable Frobenius algebra.
Wilson lines within a surface defect are then labelled by bimodules over such algebras.
However, this construction only allows for surface defects between two identical bulk theories.
A generalisation of this construction to surface defects with different bulk theories on the two sides, i.e.\ domain walls, is the main goal of this paper.

\begin{figure}
\captionsetup{format=plain, indention=0.5cm}
\centering
\begin{subfigure}[b]{0.48\textwidth}
    \centering
    \pic[1.25]{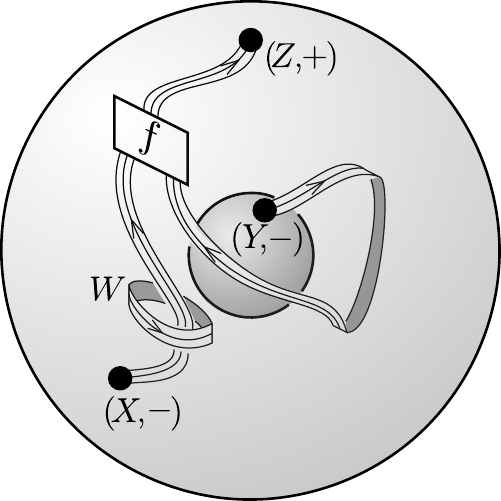}
    \caption{}
    \label{fig:Bordrib_example}
\end{subfigure}
\begin{subfigure}[b]{0.48\textwidth}
    \centering
    \pic[1.25]{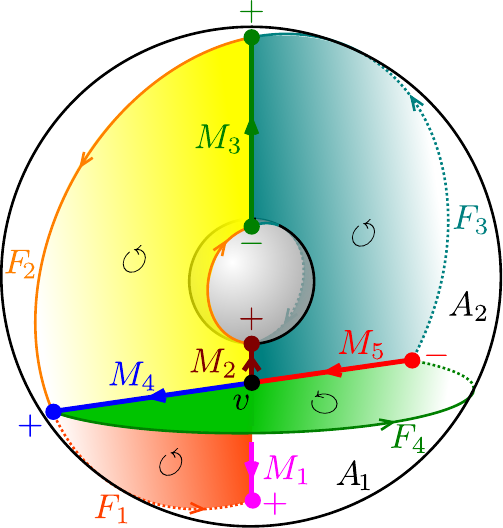}
    \caption{}
    \label{fig:Borddef_example}
\end{subfigure}
\caption{
(a) A morphism in $\Bordrib_3(\mcC)$ with underlying manifold $S^2\times[0,1]$, depicted here by a closed solid ball with an open ball removed from the interior.
The strands intersect the boundary components transversally, an intersection point, or \textit{puncture}, is framed with the tangent vector indicating the framing of the adjacent strand (not depicted in the figure) and labeled with a datum $(U,\epsilon)$, where $U\in\mcC$ labels the strand and $\epsilon=\pm$ is a sign indicating the direction of the strand ($+$ if it is directed into the boundary and $-$ otherwise).
Orientation reversal of the boundary components flips the sign, for example the inner sphere has a puncture labelled with $(Y,+)$ when seen as an incoming boundary component and $(Y,-)$ when seen as outgoing.
In pictures of bordisms $M$ we will always use the orientation on the boundary $\partial M$ as induced by the three-manifold $M$, irrespective of whether a component of the boundary is incoming or outgoing. Incoming boundaries are then parametrised by orientation-reversing embeddings.
\\
\vspace{-0.25cm}\\
(b) 
A morphism in $\Borddef_3(D_0,D_1,D_2,D_3)$ with underlying manifold $S^2\times[0,1]$ and having two $3$-strata labelled by $A_1,A_2\in D_3$, four $2$-strata labelled by $F_1,\dots,F_4\in D_2$, five $1$-strata labelled by $M_1,\dots,M_5\in D_1$ and one $0$-stratum labelled by $v\in D_0$.
The boundary components have the induced stratification, orientation reversal of a boundary component flips the orientations of all its strata.
}
\label{fig:Bord_examples}
\end{figure}

\begin{figure}
\captionsetup{format=plain, indention=0.5cm}
\pic[1.25]{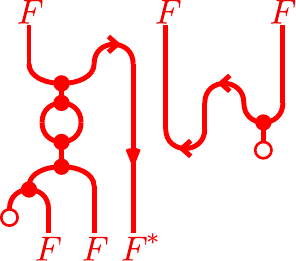} \hspace{-0.25cm}$\rightsquigarrow$\hspace{-0.25cm}
\pic[1.25]{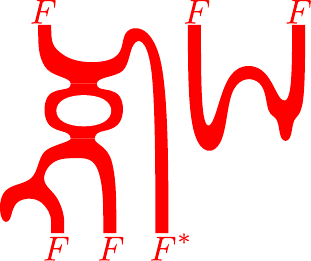} \hspace{-0.25cm}$\rightsquigarrow$\hspace{-0.25cm}
\pic[1.25]{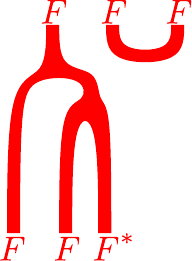} \hspace{-0.25cm}$\rightsquigarrow$\hspace{-0.25cm}
\pic[1.25]{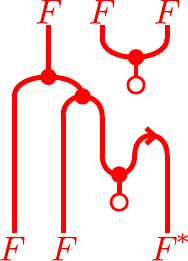}
\caption{
In graphical calculus, morphisms composed of (co)unit, (co)multiplication and (co)evaluation morphisms of a symmetric $\D$-separable Frobenius algebra $F$ can be changed by replacing the corresponding graph by a thin surface and deforming it.
$\D$-separability allows one to omit holes in the surface as indicated by the second step in the figure.
The symmetry property is used to canonically identify $F$ with its dual $F^*$.
}
\label{fig:Frob_morph}
\end{figure}

\subsubsection*{TQFT functors}
Mathematically, the RT TQFT is given by a symmetric monoidal functor
\begin{equation}
\label{eq:RT_TQFT_functor}
Z_\mcC\colon \Bordrib_3(\mcC) \ra \Vect_\opC ~,
\end{equation}
where $\Bordrib_3(\mcC)$ is a central extension of the category with $3$-dimensional bordisms with embedded ribbon graphs as morphisms and surfaces with framed $\mcC$-labelled punctures as objects.
An example of such a bordism is given in Figure~\ref{fig:Bordrib_example}.
The extension compensates a gluing anomaly (see \cite[Sec.\,IV.6]{Tu} for details); we will ignore it in the following as it does not play a role in our construction.

\medskip

In \cite{CRS2} the RT TQFT was extended to a \textit{defect TQFT}.
A $3$-dimensional defect TQFT is a symmetric monoidal functor
\begin{equation}
Z\colon \Borddef_3 (D_0, D_1, D_2, D_3) \ra \Vect
\end{equation}
where the source category is that of $3$-dimensional \it{stratified} bordisms whose $k$-strata carry labels from a set $D_k$ for each $k=0,1,2,3$ (technically the possible labels of a stratum also depend on the adjacent strata; this information together with the sets of defect labels form a so called \textit{defect datum}, see \cite{CRS1} for details and Figure~\ref{fig:Borddef_example} for an example).
Formally, \textit{defect} here means a stratum which has a label assigned.
We will also use words like \textit{point}, \textit{line}, \textit{surface} and \textit{bulk} as synonyms of $0$-, $1$-, $2$- and $3$-stratum.
Point defects are also called \textit{point insertions} while the labelled $3$-strata will also be referred to as \textit{phases} or \textit{bulk theories}.
We will call a surface defect a \textit{domain wall} if we want to stress that it can separate different bulk phases.

\medskip

In the construction of \cite{CRS2}, the label set $D_3$ for bulk phases consisted only of a single element, namely that describing the RT TQFT for a given MFC $\mcC$.
Below we describe a generalisation which includes different labels for $3$-strata.
In particular the defect datum we use has the following label sets
\begin{itemize}
\item $D_3$: commutative symmetric $\D$-separable Frobenius algebras in a fixed MFC $\mcC$.
\item $D_2$: domain walls separating two bulk theories labelled by $A,B\in D_3$ and oriented towards the one labelled with $A$ are labelled by symmetric $\D$-separable Frobenius algebras over $(A,B)$ in $\mcC$ (see Definition~\ref{def:Frob_algs_over_AB}).
We will also use the notation $D_2(A,B)$ for the collection of such algebras.
\item $D_1$: line defects in a bulk theory $A\in D_3$ are labelled by local $A$-modules; line defects having adjacent $2$-strata are labelled by \textit{multimodules} of the corresponding algebras in $D_2$ (see below).
\item $D_0$: point defects are added using the procedure of point defect completion (see \cite{CRS2} or below).
\end{itemize}
The setting of \cite{CRS2} is then recovered by restricting $D_3$ to $\{\opid_\mcC\}$.

\subsubsection*{Surface defects}
Let $F\in D_2(\opid_\mcC, \opid_\mcC)$.
Consider a morphism $f\colon F^{\otimes n} \ra F^{\otimes m}$ consisting of compositions of (co)unit and (co)multiplication as well as (co)evaluation maps.
In graphical calculus $f$ is depicted as a directed planar graph with half edges for incoming and outgoing strands.
Due to the symmetry property in \eqref{eq:sep_and_symm_properties} there is a canonical isomorphism $F\cong F^*$, so the directions of strands are of no importance
(also some of the tensor factors $F$ in the (co)domain of $f$ can be changed to $F^*$).
In fact, when dealing with such morphisms we can interpret the graph of $f$ as a surface, consisting of thin ribbon bands that branch out at the vertices.
Together, (co)unitality, (co)associativity, the Frobenius property \eqref{eq:Frob_property} and the symmetry property allow to freely deform such a surface.
The $\D$-separability property in \eqref{eq:sep_and_symm_properties} allows one in addition to fill in holes in the surface, or, equivalently, to omit one of the edges on a part of the graph which forms the boundary of an embedded disc (see Figure~\ref{fig:Frob_morph}). In particular, all connected graphs describe the same morphism $F^{\otimes n} \ra F^{\otimes m}$.

\medskip

In RT TQFT, a surface defect $S$ between identical phases is evaluated by replacing it with an $F$-labelled graph $\Gamma_S\subseteq S$.
This graph is required to be \textit{full}, that is, any new edge added to $\Gamma_S$ which again lies in $S$ can be removed again by using the $\D$-separability property.
Such a full graph can be obtained for example by removing from each connected component of $S$ a non-zero (but finite) number of open discs and taking the deformation retract.
The strands of $\Gamma_S$ have a framing with a normal vector that is tangent to the surface $S$ and together with the orientation of the strand coincides with the orientation of $S$.

\subsubsection*{Bulk phases}
For $A\in D_3$ one can apply an analogous procedure to evaluate bulk phases. 
Namely one embeds a full $A$-labelled graph $\Gamma_U$ into a $3$-stratum $U$.
To interpret $\Gamma_U$ as a morphism in $\mcC$ we again note that the edges of $\Gamma_U$ need not be directed and that the vertices can be interpreted as morphisms obtained by composing (co)unit, (co)multiplication, (co)evaluation and the braiding morphisms.
We note that the framing of the $A$-labelled strands is not important, since $A\in\mcC$ has a trivial twist:  $\theta_A = \id_A$.
\begin{rem}\label{rem:A-phase-is-RT-phase}
Evaluating a manifold $M$ which has no surface defects and whose $3$-strata are labelled by a fixed commutative haploid symmetric $\D$-separable Frobenius algebra $A$, one gets precisely the invariant $Z_{\mcC_A^\loc}(M)$ 
\cite{CMRSS1,CMRSS2}.
That is, the invariant which is assigned to $M$ by the RT TQFT obtained from the modular category of local modules of $A$, cf.\ Theorem~\ref{thm:CAloc_is_modular}.
\end{rem}

\subsubsection*{Domain walls between two bulk phases}
We now consider two $3$-strata $U$ and $V$ labelled by algebras $A,B\in D_3$ adjacent to a surface defect $S$ labelled by $F\in D_2(A,B)$.
The evaluation procedure is very similar to before: $S$ is replaced by a full planar graph $\Gamma_S$ labelled by $F$ and $U$, $V$ by full graphs $\Gamma_U$, $\Gamma_V$ labelled by $A$ and $B$.
The graphs $\Gamma_U$ and $\Gamma_V$ have edges which are adjacent to $S$.
We move their ends so that they lie on $\Gamma_S \subseteq S$.
The whole configuration of defects is now replaced by the graph $\Gamma_U \cup \Gamma_V \cup \Gamma_S$ where the points in $\Gamma_U \cap \Gamma_S$ are labelled by (co)actions of $A$ on $F$ and those in $\Gamma_V \cap \Gamma_S$ by (co)actions of $B$ on $F$.

\medskip

Note that having an $F\in D_2(A,B)$ labelled surface defect oriented towards $U$ gives the same invariant as having an $F^{\operatorname{op}}\in D_2(B,A)$ labelled surface defect oriented towards $V$ (see Remark~\ref{rem:Frob_algs_over_AB_cross_opp}).
Indeed, consistently rotating all $F$-strands and coupons of $\Gamma_S$, one obtains crossings as in the definition of (co)multiplication and $A$-$B$-actions of $F^{\operatorname{op}}$.
\begin{figure}
\captionsetup{format=plain, indention=0.5cm}
\centering
\begin{subfigure}[b]{0.45\textwidth}
    \centering
    \pic[1.25]{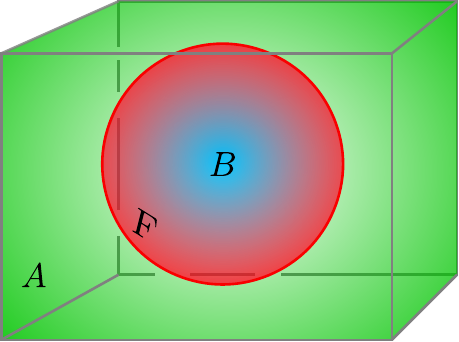}
    \caption{}
    \label{fig:example_S2_in_S3_bord}
\end{subfigure}
\begin{subfigure}[b]{0.45\textwidth}
    \centering
    \pic[1.25]{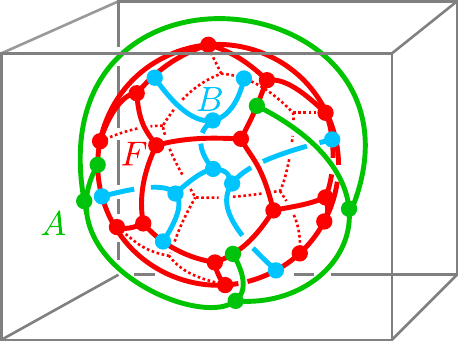}
    \caption{}
    \label{fig:example_S2_in_S3_graph}
\end{subfigure}
\caption{
(a) The stratified sphere $S^3_{\operatorname{def}}$.
The boundary of the cube represents a single point in $S^3 \simeq \opR^3 \cup \{\infty\}$.
(b) An example of a ribbon graph corresponding to the stratification of $S^3_{\operatorname{def}}$.
The directions of the strands are omitted, the vertices where $A$, $B$ and $F$ lines meet are labelled by (co)multiplications or (co)actions.
}
\label{fig:example_S2_in_S3}
\end{figure}
\begin{example}
Let $S^3_{\operatorname{def}}$ be a stratified $3$-sphere with two $3$-strata given by two open $3$-balls and an $S^2$-surface between them, labelled by $A,B\in D_3$ and $F\in D_2(A,B)$ as depicted in Figure~\ref{fig:example_S2_in_S3_bord}.
Upon evaluation one replaces the $2$- and $3$-strata with a ribbon graph in $S^3$
as shown in Figure~\ref{fig:example_S2_in_S3_graph}.
One can use the relations in Definition~\ref{def:Frob_algs_over_AB} as well as properties Frobenius algebras to further simplify the graph: all $A$- and $B$-lines can be contracted as well as all but one $F$-line.
This leaves the unstratified  sphere $S^3$ with a single $F$-labelled loop embedded.
The corresponding invariant is therefore the scalar:
\begin{equation}
    Z_\mcC(S^3_{\operatorname{def}}) = \dim F \cdot Z_\mcC(S^3) ~.
\end{equation}
\end{example}
\begin{rem}
In Remark~\ref{rem:A-phase-is-RT-phase} we have seen that a bulk phase labelled $A$ describes the RT TQFT for the modular category $\mcC_A^\loc$. 
Domain walls labelled by $F \in D(A,B)$ can therefore be thought of as separating the RT TQFTs for modular categories $\mcC_A^\loc$ and $\mcC_B^\loc$.
We will recall in Section~\ref{subsec:Witt_eq} that $\mcC_A^\loc$ and $\mcC_B^\loc$ are Witt-equivalent.
Conversely, given two modular categories $\mcC$ and $\mcD$ which are Witt equivalent, there exist a (non-unique) modular category $\mcE$ together with commutative symmetric $\Delta$-separable Frobenius algebras $A,B \in \mcE$, such that $\mcC \cong \mcE_A^\loc$ and $\mcD \cong \mcE_B^\loc$. 
In this sense, our construction allows one to describe domain walls between RT-TQFTs for arbitrary Witt-equivalent modular categories.
In Section~\ref{subsec:domain_walls_btw_RT_theories} we recall the argument of \cite{FSV} that non-Witt-equivalent theories cannot be separated by a (topological) domain wall.
The construction can be extended to any finite number $\mcC_1,\dots,\mcC_n$ of modular categories. However, there is no single modular category $\mcE$ which would allow one to obtain all other Witt-equivalent modular categories in terms of local modules.
\end{rem}

\subsubsection*{Wilson lines in bulk phases and surfaces}
The construction of invariants of stratified manifolds with domain walls only has a straightforward generalisation which allows one to evaluate Wilson lines both in the bulk and inside surface defects:
\begin{itemize}
\item 
Wilson lines in a $3$-stratum $U$ labelled $A\in D_3$ are labelled with local $A$-modules.
Upon evaluation one again attaches the graph $\Gamma_U$ to the network of Wilson lines in $U$ using (co)action morphisms.
\item
A Wilson line in a surface defect can separate two $2$-strata labelled with $F_1, F_2 \in D_2(A,B)$, where $A,B\in D_3$ label the adjacent bulk theories.
Accordingly, it is labelled by an $F_1$-$F_2$-bimodule $M$ over $(A,B)$ is in Definition~\ref{def:bimodules_over_AB}.
A network of Wilson lines in a surface defect is then evaluated by ``attaching'' the graphs of adjacent $2$-strata using (co)action morphisms.
\end{itemize}

\begin{figure}
\captionsetup{format=plain, indention=0.5cm}
\centering
\pic[1.25]{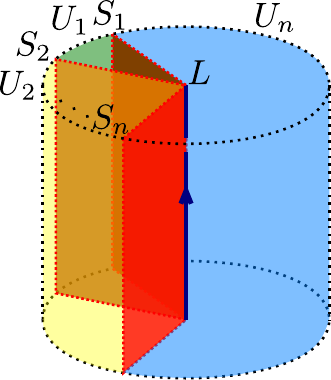}
\caption{
Neighbourhood of a line defect $L$ with several domain walls attached.
}
\label{fig:local_neigh_3d_line}
\end{figure}
\subsubsection*{General line defects}
A general line defect $L$ can have $n>0$ adjacent $2$-strata $S_1,\dots,S_n$ and $3$-strata $U_1,\dots,U_n$.
Let us choose $S_1$ as a preferred $2$-stratum, so that a tubular neighbourhood of $L$ looks like in Figure~\ref{fig:local_neigh_3d_line}.
For $i=1,\dots,n$, let $A_i\in D_3$ be the algebras labelling $U_i$ and $F_i\in D_2(A_i,A_{i-1})$ the algebras labelling $S_i$, where we use the convention $A_0 := A_n$ (using the opposite algebras one can without loss of generality assume that $S_1,\dots,S_n$ have paper plane orientation).
To label the line $L$ we make the following (cf.\ \cite[Def.\,.2.2]{CRS2})
\begin{defn}
\label{def:multimodule}
An \textit{$F_1,\dots,F_n$-multimodule} is an object $M\in\mcC$ which is simultaneously a module of the algebras $F_1,\dots,F_n$ such that
i) the identity
\begin{equation}
\pic[1.25]{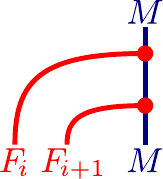} = \pic[1.25]{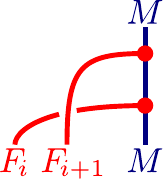}
\end{equation}
holds for all $i=1,\dots,n$ and ii) the two $A_i$-actions obtained from $F_i$- and $F_{i+1}$-actions by analogy to \eqref{eq:AB-action-from-F-action-1}, \eqref{eq:AB-action-from-F-action-2} are equal (cf.\ Definition~\ref{def:bimodules_over_AB}).
\end{defn}
Upon evaluation, the $A_i$-labelled graphs $\Gamma_{U_i}$ and $F_i$-labelled graphs $\Gamma_{S_i}$ are attached to the line $L$ labelled by an $F_1,\dots,F_n$-multimodule using (co)action morphisms in a similar manner as before.

\medskip

Note that the choice of a preferred adjacent $2$-stratum $S_1$ gives $L$ a natural framing which is necessary to label it with an object $M\in\mcC$.
In \cite{CRS2} it was shown how to deal with unframed line defects by equipping $M$ with a so called \textit{cyclic structure}.
Although the setting in \cite{CRS2} did not include different labels for the adjacent $3$-strata (which in our setting corresponds to $A_1 = \dots = A_n = \opid_\mcC$), the same idea can be adapted here; we omit the details.

\subsubsection*{State spaces}

Until now we only explained how to compute the numerical invariants that the TQFT assigns to closed $3$-manifolds.
To obtain linear maps for $3$-manifolds with boundary, one needs to know also the vector spaces assigned to surfaces.
A boundary surface $\Sigma$ of a stratified manifold $M$ is assumed to intersect the line and surface defects transversally.
$\Sigma$ itself is therefore stratified, its $0$- and $1$-strata carry orientations and labels which are induced by the corresponding $1$- and $2$-strata of $M$.

\medskip

The evaluation procedure for a stratified bordism $M\colon \Sigma \to \Sigma'$ is similar to that for a closed manifold: the $2$- and $3$-strata are replaced by full graphs which are labelled by the corresponding Frobenius algebras.
The difference now is that the graph $\Gamma_D$, where $D$ is either a $2$- or $3$-stratum, can have edges which intersect the boundary $(-\Sigma) \sqcup \Sigma'$ of $M$.
That $\Gamma_D$ is \textit{full} in this case means the following: Each new edge that is added to $\Gamma_D$ and which lies in $D$ and does not intersect the boundary can be again be omitted without changing the result after applying $Z_\mcC$.

The intersection points of $\Gamma_D$ with the boundary of $M$ are treated as punctures on the underlying unstratified surfaces for $\Sigma$ and $\Sigma'$.
Note that different choices for $\Gamma_D$ result in different number of punctures, but in the interior this choice does not matter: graphs $\Gamma_D$ and $\Gamma_D'$ yielding the same punctures can be modified into each other by deforming and adding/removing edges in the interior.

\medskip

The dependence on the number of punctures on the boundary can be eliminated by the standard limit procedure:
Consider the cylinder $C = \Sigma\times[0,1]$.
Converting its $2$- and $3$-strata to ribbon graphs one obtains a cylinder bordism $C_p^{p'}\colon\Sigma_p \ra \Sigma_{p'}$, where $p$, $p'$ denote the corresponding punctures on the incoming and outgoing boundaries.
The linear maps $\Psi_p^{p'} := Z_\mcC(C_p^{p'})$ then satisfy $\Psi_{p'}^{p''}\circ\Psi_p^{p'} = \Psi_p^{p''}$, i.e.\ form a directed system.
The state space of the stratified surface $\Sigma$ is then defined as its colimit.
Note that $\Psi_p^p$ is an idempotent and that the state spaces can be computed as
\begin{equation}
\label{eq:statesp}
Z_\mcC(\Sigma) \cong \im \Psi_p^p ~.
\end{equation}

\begin{figure}
\captionsetup{format=plain, indention=0.5cm}
\centering
\begin{subfigure}[b]{0.32\textwidth}
    \centering
    \pic[1.25]{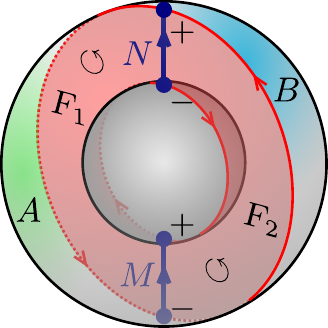}
    \caption{}
    \label{fig:example_S2ss_cylinder}
\end{subfigure}
\begin{subfigure}[b]{0.32\textwidth}
    \centering
    \pic[1.25]{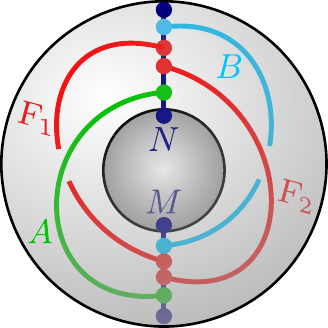}
    \caption{}
    \label{fig:example_S2ss_cylinder_rib}
\end{subfigure}
\begin{subfigure}[b]{0.32\textwidth}
    \centering
    \pic[1.25]{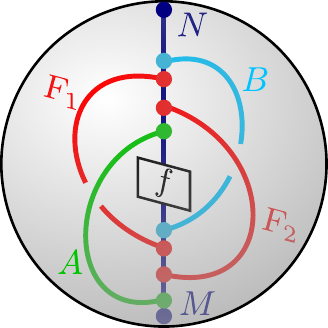}
    \caption{}
    \label{fig:example_S2ss_cylinder_rib_img}
\end{subfigure}
\caption{
(a) The stratified cylinder over $S^2_\text{def}$ in Example~\ref{ex:state-space-S2-def} whose image defines the state space. (b) A representation of the bulk phases and domain walls in terms of ribbons. (c) Action of the cylinder on a generic state in the underlying RT TQFT, given by a solid ball with a morphism $f \in \mcC(M,N)$ inserted.
}
\label{fig:example_S2ss}
\end{figure}

\begin{example}\label{ex:state-space-S2-def}
Let $A,B\in D_3$ and $F_1,F_2\in D_2(A,B)$.
Consider a stratified $2$-sphere $S^2_{\operatorname{def}}$ with two lines labelled by $F_1$ and $F_2$ adjacent to two points labelled by $F_1$-$F_2$-bimodules $M$ and $N$ over $(A,B)$.
The cylinder $C = S^2_{\operatorname{def}} \times [0,1]$  and a choice for its ribbonisation $C_p^p$ are depicted in Figures~\ref{fig:example_S2ss_cylinder} and \ref{fig:example_S2ss_cylinder_rib}.
In this case $p$ consists of two punctures labelled by objects $M,N\in\mcC$.
Let $S^2_p$ be the corresponding $2$-sphere with punctures so that $[C_p^p\colon S^2_p\ra S^2_p]\in\Bordrib_3(\mcC)$.
By construction of the RT TQFT, we have an isomorphism $Z_\mcC(S^2_p) \cong \mcC(M,N)$ of vector spaces.
For $f\in\mcC(M,N)$, the image $\Psi_p^p(f)$ corresponds to the bordism $\varnothing\ra S^2_p$ depicted in Figure~\ref{fig:example_S2ss_cylinder_rib_img}.
The $F_1$ and $F_2$ lines then act as a projector on the subspace of $F_1$-$F_2$-bimodule morphisms $M\ra N$ (see e.g.\ \cite[Lem.\,4.4]{FRS1}), while the $A$ and $B$ lines can be contracted and do not give additional conditions.
From \eqref{eq:statesp} one therefore gets
\begin{equation}
    Z_\mcC(S^2_{\operatorname{def}}) \cong {}_{F_1}\mcC_{F_2}(M,N)
\end{equation}
in terms of homomorphisms of $F_1$-$F_2$-bimodules.
\end{example}

\begin{figure}[t]
\captionsetup{format=plain, indention=0.5cm}
\centering
\begin{subfigure}[b]{0.48\textwidth}
    \centering
    \pic[1.25]{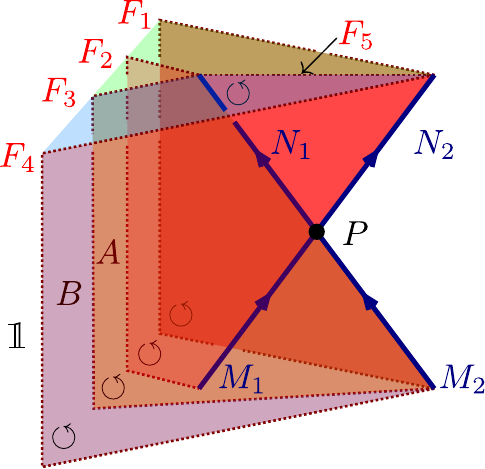}
    \caption{}
    \label{fig:point_defect_neigh}
\end{subfigure}
\begin{subfigure}[b]{0.48\textwidth}
    \centering
    $S^2_P =$\pic[1.25]{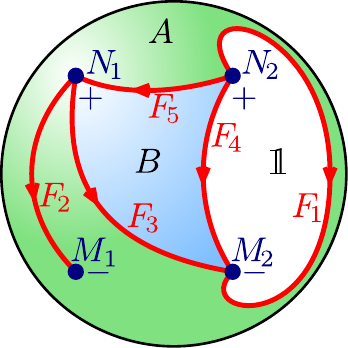}
    \vspace{0.75cm}
    \caption{}
    \label{fig:point_defect_sphere}
\end{subfigure}
\caption{
(a) An example of a $0$-stratum $P$ having three adjacent bulk phases
(labelled by $A,B,\opid_\mcC\in D_3$), five surface defects (labelled by $F_1,\dots,F_5$, so that $F_1\in D_2(A,\opid)$, $F_2\in D_2(A,A)$ , etc.) and four line defects (labelled by multimodules $M_1$, $M_2$, $N_1$, $N_2$, so that $M_1$ is a left $F_2$-module, $M_2$ is an $F_1,F_3,F_4$ multimodule, etc.).
(b) The boundary of a closed ball neighborhood of $P$.
It is a $2$-sphere with stratification and labels induced by the strata adjacent to $P$.
The set of possible labels of $P$ is defined to be the state space $Z_\mcC(S^2_P)$.
}
\label{fig:point_defect}
\end{figure}
\subsubsection*{Point defects}
A $0$-stratum $P$ can have several adjacent $1$-, $2$- and $3$-strata.
The label set for $P$ is given by the state space $Z_\mcC(S^2_P)$, where $S^2_P$ is a stratified $2$-sphere obtained as a boundary component after removing a small open ball $B_P$ surrounding $P$\footnote{This applies to point defects with a parametrised neighbourhood. In the unparametrised case one has to divide out symmetries, see \cite[Sec.\,2.4]{CRS1} for details.}.
The evaluation procedure is then as follows:
Let $M$ be a closed stratified manifold with point defects $P_1,\dots,P_n$ labelled by $v_i\in Z_\mcC(S^2_{P_i})$, $i=1,\dots,n$.
Denoting $M_\circ := M \setminus (B_{P_1} \sqcup \dots \sqcup B_{P_n})$, one gets a linear map
\begin{equation}
Z_\mcC(M_\circ)\colon Z_\mcC(S^2_{P_1}) \otimes \cdots \otimes Z_\mcC(S^2_{P_n}) \ra \opC ~.
\end{equation}
The invariant of $M$ is then defined by
\begin{equation}
Z_\mcC(M) := Z_\mcC(M_\circ) (v_1 \otimes \cdots \otimes v_n) ~.
\end{equation}
In practice, the spaces $Z_\mcC(S^2_P)$ turn out to be given by spaces of multimodule morphisms.
For example, in the situation depicted in Figure~\ref{fig:point_defect_neigh}, the space consists of morphisms 
$f\colon M_1 \otimes M_2 \ra N_1 \otimes N_2$ such that
\begin{equation}
\pic[1.25]{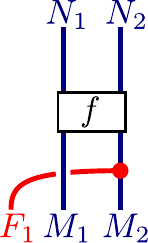} \hspace{-0.25cm}=\hspace{-0.25cm} \pic[1.25]{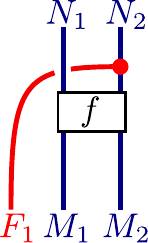} ,
\pic[1.25]{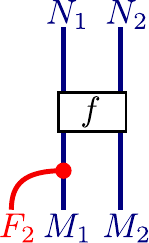} \hspace{-0.25cm}=\hspace{-0.25cm} \pic[1.25]{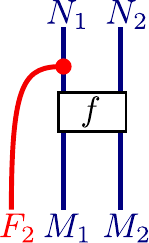} ,
\pic[1.25]{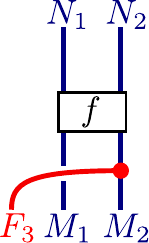} \hspace{-0.25cm}=\hspace{-0.25cm} \pic[1.25]{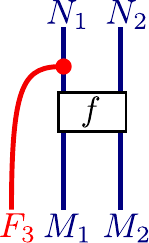} , \text{ etc.}
\end{equation}
Finally, we mention that such a procedure of adding point defects works for any construction of a defect TQFT and has been called ``point defect completion'' in \cite[Sec.\,2.4]{CRS1}.

\subsection{Bicategories of domain walls}
\label{subsec:bicats_of_domain_walls}
In this section we look at some aspects of the higher categorical description of defect TQFTs.
Recall that a \textit{bicategory} $\mcB$ consists of 
\begin{itemize}
\item a collection of \textit{objects}, which we will also call $\mcB$ by abuse of notation,
\item categories $\mcB(A,B)$ for each pair of objects $A,B\in\mcB$, whose objects are called \textit{$1$-morphisms} morphisms are in turn called \textit{$2$-morphisms} of $\mcB$,   
\item \textit{composition functors} $\otimes\colon\mcB(B,C)\times\mcB(A,B) \ra \mcB(A,C)$ for all $A,B,C\in\mcB$, together with associativity isomorphisms.
\item a \textit{unit} for each $\mcB(A,A)$, together with unit isomorphisms for the composition.
\end{itemize}
For more details and for the axioms these data have to satisfy, see e.g.\ \cite{Le}.
Bicategories have a version of graphical calculus, in which strands and coupons are labelled with $1$- and $2$-morphisms while the $2$-dimensional patches between them are labelled with the corresponding objects.
In particular, the datum of a bicategory with one object is equivalent to a monoidal category.
One can also define structures on bicategories which allow one to make their graphical calculus richer, for example in a \textit{pivotal bicategory} each $1$-morphism $M\in\mcB(A,B)$ has a two-sided dual (or biadjoint) $M^*\in\mcB(B,A)$ together with (co)evaluation $2$-morphisms satisfying analogous conditions to those in \eqref{eq:pivotal_conds},
see \cite[\S 2]{KSt}, \cite[Sec.\,2.2]{Ca}.

\medskip

A functor between bicategories $\mcB$, $\mcB'$ is a mapping between objects $F\colon\mcB\ra\mcB'$ together with a collection of functors $\{ \mcB(A,B)\ra\mcB'(FA,FB) \}_{A,B\in\mcB}$ equipped with natural transformations, analogous to the monoidal structure on a functor between two monoidal categories.
As with 1-categories, we say that a functor $F\colon\mcB\ra\mcB'$ is an \textit{equivalence} if there is a functor in the opposite direction such that the two compositions are isomorphic to identity functors. Again analogous to 1-categories, $F$ is an equivalence iff the corresponding functors on the categories of $1$-morphisms are equivalences (i.e.\ $F$ is fully faithful) and each object of $\mcB'$ has an invertible $1$-morphism to an object in the image of $F$ (i.e.\ $F$ is essentially surjective), cf.\ \cite{Le}.

\medskip

Let $Z\colon \Borddef_3 (D_0, D_1, D_2, D_3) \ra \Vect$ be a $3$-dimensional defect TQFT.
To each pair of labels $u,v\in D_3$ of $3$-strata one can assign a pivotal bicategory $\mcB_{u,v}$ which has (see \cite[Sec.\,3.3]{CMS}, \cite[Sec.\,2.4]{DKR}):
\begin{itemize}
\item \textit{objects:}
labels for $2$-strata between two $3$-strata labelled by $u$, $v$, oriented so that the normal vector points towards $u$;
\item \textit{$1$-morphisms:}
(lists of)\footnote{
    Using lists of 1-strata-labels as 1-morphisms allows one to define composition as concatenation.
    We will ignore this technicality in what follows.
}
labels for $1$-strata having two adjacent $2$-strata labelled with objects in $\mcB_{u,v}$;
\item {\textit{$2$-morphisms:}}
labels for $0$-strata whose adjacent $1$-strata are $1$-morphisms of $\mcB_{u,v}$.
\end{itemize}
\begin{rem}
The bicategory of surface defects from above belongs to a more general construction assigning a \textit{tricategory} (i.e.\ a $3$-dimensional analogue of a bicategory) to a $3$-dimensional defect TQFT, see \cite{CMS}.
\end{rem}

\medskip

Let us apply the construction of the bicategory of surface defects to the TQFT $Z_\mcC$ defined in Section~\ref{subsec:constr_of_defect_TQFT}.
Assuming that the adjacent bulk theories are labelled with algebras $A,B\in D_3$, it readily yields the bicategory in the following
\begin{defn}
\label{def:FrobAlg-over-AB}
Let $A,B\in\mcC$ be commutative symmetric $\D$-separable Frobenius algebras.
We denote by $\FrobAlg^{\operatorname{s},\D}_{\mcC,A,B}$ the bicategory 
\begin{itemize}
\item having symmetric $\D$-separable Frobenius algebras over $(A,B)$ in $\mcC$ as objects,
\item $\FrobAlg^{\operatorname{s},\D}_{\mcC,A,B}(F_1,F_2)$ being the category of $F_1$-$F_2$-bimodules over $(A,B)$ and their morphisms,
\item the composition of two bimodules ${}_{F_1}M_{F_2}$ and ${}_{F_2} M_{F_3}$ being the tensor product $M \otimes_{F_2} N$ over the respective algebra,
\item for each object $F$, the unit being $F$ seen as a bimodule over itself.
\end{itemize}
\end{defn}

\section{Domain walls via module categories}
\label{sec:3}

\subsection{Module categories}
\label{subsec:module_cats}
Let $\mcA$ be a multifusion category. By a (left) $\mcA$-module category we will mean a finitely semisimple $\mathbb{C}$-linear category $\mcM$ together with an action $\triangleright\colon\mcA \times \mcM \ra \mcM$ and natural isomorphisms $(-\otimes -)\triangleright - \Ra - \triangleright (- \triangleright -)$ and $\opid_\mcC \triangleright - \Ra \Id_\mcM$ satisfying the usual pentagon and triangle identities.
An $\mcA$-module functor between $\mcA$-module categories $\mcM$ and $\mcN$ is a linear functor $F\colon\mcM\ra\mcN$ equipped with natural isomorphism $F(- \triangleright -) \Ra - \triangleright F(-)$ satisfying compatibility conditions.
Natural transformations of $\mcA$-module functors $F,G\colon\mcM\ra\mcN$ are assumed to commute with the structure morphisms of $F$ and $G$.
A module category is called indecomposable if it is not equivalent (as a module category) to a direct sum of two non-trivial module categories.
Detailed definitions can be found in \cite[Ch.\,7]{EGNO}.\footnote{
    We stress again that we assume all module categories to be \textsl{semisimple}. This is only a special case of the exact module categories considered in \cite{EGNO} but sufficient for our purpose.}

\medskip

Let $A\in\mcA$ be an algebra.
Then the category $\mcA_A$ of \textit{right} $A$-modules is a \textit{left} $\mcA$-module category with the action $X \triangleright (M, \rho) := (X \otimes M, \widetilde\rho)$, where the right $A$-action is given by $\widetilde\rho = \big[ (X \otimes M) \otimes A \xrightarrow{\sim} X \otimes (M \otimes A) \xrightarrow{\id_X \otimes \rho} X \otimes M \big]$.
Two algebras are said to be \textit{Morita equivalent} if their respective categories of right modules are equivalent as $\mcA$-module categories.
The algebra $A$ is called semisimple if $\mcA_A$ is semisimple.
One has by \cite{Os}, and e.g.\ \cite[Cor.\,7.10.5.(i)]{EGNO}:

\begin{prp}
\label{prp:modcats-algs}
Let $\mcM$ be an $\mcA$-module category.
Then there exists a semisimple algebra $A\in\mcA$ such that $\mcM \simeq \mcA_A$ as module categories.
\end{prp}

\medskip

Let $\mcA$-$\operatorname{Mod}$ be the bicategory of $\mcA$-module categories, functors and natural transformations and $\operatorname{Alg}_\mcA$ be the bicategory of semisimple algebras in $\mcA$, their bimodules and bimodule morphisms.
Building on Proposition~\ref{prp:modcats-algs} one gets:
\begin{prp}
\label{prp:Alg_Mod_equiv}
Let $A,B\in\mcA$ be semisimple algebras, $M,N\in\mcA$ be $A$-$B$-bimodules and $[f\colon M\ra N]$ a bimodule morphism.
The functor $\operatorname{Alg}_\mcA$ $\ra$ $\mcA$-$\operatorname{Mod}$ defined on
\begin{equation}
\label{eq:alg_modcat_functor}
\begin{array}{llll}
\text{objects:} &  A  & \mapsto & \mcA_A\\
\text{1-morphisms:} & M & \mapsto & [- \otimes_A M\colon \mcA_A \ra \mcA_B]\\
\text{2-morphisms:} & f & \mapsto & \{\id_L \otimes_A f\}_{L \in \mcA_A}
\end{array}
\end{equation}
is an equivalence of bicategories.
\end{prp}
The following is a convenient criterion to compare the module categories of two fusion categories $\mcA$ and $\mcB$, see \cite{Mu1} and \cite[Thm.\,3.1]{ENO2}:
\begin{prp}
\label{prp:ZA_eq_ZB_means_Morita_eq}
The bicategories $\mcA$-$\operatorname{Mod}$ and $\mcB$-$\operatorname{Mod}$ are equivalent if and only if $\mcZ(\mcA) \simeq \mcZ(\mcB)$ as braided fusion categories.
\end{prp}

\medskip

Let $\mcC$ be a pivotal fusion category and let $A\in\mcC$ be a symmetric $\D$-separable Frobenius algebra.
The question what extra structure the $\mcC$-module category $\mcC_A$ has in this situation was addressed in \cite{Schm}, where the following notion was introduced:

\begin{defn}
A \textit{module trace} on a module category $\mcM$ of a pivotal fusion category $\mcC$ is a collection of linear maps $\Theta_M\colon \End_\mcM M \ra \opC$, $M\in\mcM$, such that for all $X\in\mcC$ and $M,N\in\mcM$:
\begin{enumerate}[i)]
\item One has
\begin{equation}
\Theta_M(g \circ f) = \Theta_N(f \circ g)
\end{equation}
for all $f \in \mcM(M,N)$ and $g \in \mcM(N,M)$.
\item The following pairing is non-degenerate:
\begin{equation}
\label{eq:mod_trace_def_omega}
\omega_{M,N}\colon \mcM(M,N) \otimes_\opC \mcM(N,M) \ra \opC, \quad
f \otimes_\opC g \mapsto \Theta_M(g \circ f) \, .
\end{equation}
\item 
For 
$f\in\End_\mcM (X \triangleright M)$ let $\overline{f}\colon M \ra M$ be given by
$$
\overline{f} :=
(\ev_X \triangleright \id_M)\circ(\id_{X^*} \triangleright f)\circ(\coevt_X \triangleright \id_M)
$$
(we have omitted coherence isomorphisms for readability).
Then
\begin{equation}
\Theta_{X\triangleright M}(f) = \Theta_M ( \overline{f} ).
\end{equation}
\end{enumerate}
\end{defn}
Module traces satisfy the following uniqueness property \cite[Prop.\,4.4]{Schm}:
\begin{prp} \label{prp:mod_trace_unique}
If an indecomposable module category $\mcM$ has a module trace $\Theta$, then any other module trace $\Theta'$ will be proportional to $\Theta$, i.e.\ $\Theta' = z\cdot\Theta$, $z\in\opC^\times$.
\end{prp}

The module category $\mcC_A$ has a module trace $\Theta_M(f) := \tr_l f$.
Combining Proposition~\ref{prp:modcats-algs} with \cite[Sec.\,6]{Schm}, one has: 
\begin{prp}
\label{prp:mod_trace_vs_Frob_algs}
Let $\mcM$ be a $\mcC$-module category with module trace.
Then there exists a symmetric $\D$-separable Frobenius algebra $A$, such that $\mcM \simeq \mcC_A$ as module categories.
\end{prp}
\begin{defn}
\label{def:Modtr_FrobalgCD}
We let $\mcC$-$\operatorname{Mod}^{\tr}$ be the bicategory of $\mcC$-module categories with module trace, module functors and natural transformations and $\FrobAlg_\mcC^{\operatorname{s},\D}$ be the bicategory of symmetric $\D$-separable Frobenius algebras in $\mcC$, their bimodules and bimodule morphisms.
\end{defn}
Note that in the definition of $\mcC$-$\operatorname{Mod}^{\tr}$ we do not require the module functors to be compatible with module traces (so-called \textit{isometric functors}, see \cite[Def.\,3.10]{Schm}).
By Propositions~\ref{prp:Alg_Mod_equiv} and \ref{prp:mod_trace_vs_Frob_algs} we get
\begin{prp}
\label{prp:FrobalgCD_equiv_Modtr}
The functor $\FrobAlg_\mcC^{\operatorname{s},\D} \ra \mcC$-$\operatorname{Mod}^{\tr}$ defined analogously as in \eqref{eq:alg_modcat_functor} is an equivalence of bicategories.
\end{prp}

\subsection{Domain walls between Reshetikhin-Turaev theories}
\label{subsec:domain_walls_btw_RT_theories}
Recall from Section~\ref{subsec:bicats_of_domain_walls} that surface defects between two bulk theories in a $3$-dimensional defect TQFT can be collected into a bicategory which has such surface defects as objects, Wilson lines separating two surface defects as $1$-morphisms and point insertions as $2$-morphisms.
The bicategory of surface defects between two bulk theories of RT type was studied from this point of view in \cite{FSV}.
The surface defects in question were only considered locally, i.e.\ not as part of a stratification of a compact manifold and without providing a construction of a defect TQFT.
Indeed, one can define a natural bicategory of surface defects separating two bulk theories labelled by MFCs $\mcC$ and $\mcD$, even before an exact construction of a complete defect TQFT is known.
Concretely, the algebraic description given in \cite{FSV} is:
Let $\mcW$ be a pivotal fusion category such that there is a braided equivalence
\begin{equation}
\label{eq:Witt_triv_motivation}
F\colon \mcC \boxtimes \widetilde{\mcD} \xra{\sim} \mcZ(\mcW) \, .
\end{equation}
Then the bicategory of surface defects is $\mcW$-$\operatorname{Mod}$.
A defect between the two bulk theories can therefore only exist if the modular categories $\mcC$ and $\mcD$ describing them are Witt equivalent (see Definition~\ref{def:Witt_eq} below).

\medskip

Let us review the argument of \cite{FSV}.
Having a surface defect $S$, the labels for Wilson lines within it form a category $\mcW$.
The topological nature of Wilson lines implies that $\mcW$ is monoidal and pivotal; we also assume it to be fusion.
Having an $X\in\mcC$ labelled Wilson line in the bulk, one can adiabatically bring it next to $S$ so that it becomes a defect Wilson line $F_\ra(X)\in\mcW$.
Since it is merely ``hovering'' next to $S$, it can cross to the other side of any defect Wilson line $W\in\mcW$, i.e.\ one has a family of morphisms $F_\ra(X) \otimes W \ra W \otimes F_\ra(X)$, which assemble into a half-braiding for $F_\ra(X)$.
This implies the existence of a functor of braided categories $F_\ra\colon\mcC \ra \mcZ(\mcW)$.
An analogous argument then gives a functor $F_\la\colon \widetilde{\mcD} \ra \mcZ(\mcW)$.
For $X\in\mcC$, $Y\in\mcD$ we then define the functor in \eqref{eq:Witt_triv_motivation} by $F(X\otimes Y) := F_\ra(X) \otimes F_\la(Y)$ and assume it to be an equivalence.

One can then consider two parallel defect Wilson lines.
One of them has defect condition $S$ to both sides.
According to the preceding discussion, it is labelled by an object $W\in\mcW$.
The other defect Wilson line separates defect conditions $S$ and $S'$ and is labelled by an object $W'$ of a category $\mcW_{S'}$.
Fusing the two Wilson lines yields a new Wilson line which must be labelled by an object $W\triangleright W' \in \mcW_{S'}$.
Since this should also be compatible with pointlike insertions on the Wilson lines, we get an action $\mcW \times \mcW_{S'}\to \mcW_{S'}$, so that $\mcW_{S'}$ gets the structure of a $\mcW$-module category.
Thus, the bicategory of defect conditions is equivalent to the bicategory of $\mcW$-module categories.

\medskip

Let $S$ and $S'$ be two surface defects as in the setting above with the corresponding categories $\mcW$ and $\mcV$ of surface Wilson lines.
Note that by the argument above one must have a braided equivalence $\mcZ(\mcW) \simeq \mcZ(\mcV)$.
By Proposition~\ref{prp:ZA_eq_ZB_means_Morita_eq} this implies that the bicategories $\mcW$-$\operatorname{Mod}$, $\mcV$-$\operatorname{Mod}$ of module categories of $\mcW$ and $\mcV$ are equivalent.
The choice of $\mcW$ and the equivalence as in \eqref{eq:Witt_triv_motivation} therefore serves as a choice of ``coordinates'' which help describing the abstract bicategory of surface defects in a more concrete way.

\subsection{Module traces from sphere defects}
\label{subsec:mod_traces_from_spheres}
Let us now extend the treatment in \cite{FSV} by considering a hypothetical defect TQFT $Z$, whose bulk theories are of RT type, whose defects are described as in Section~\ref{subsec:domain_walls_btw_RT_theories}.
As we have seen, surface defects separating theories labelled by $\mcC$ and $\mcD$ are described by module categories of a fusion category $\mcW$, for which there is an equivalence $\mcC \boxtimes \widetilde{\mcD} \simeq \mcZ(\mcW)$.
We show that under a reasonable assumption on $Z$, the module categories in question are of a particular type, namely they have a module trace.

\medskip

\begin{figure}
\captionsetup{format=plain, indention=0.5cm}
\centering
\begin{subfigure}[b]{0.45\textwidth}
    \centering
    \pic[1.25]{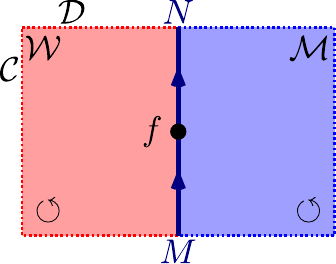}
    \caption{}
    \label{fig:lines_on_def_morphlab}
\end{subfigure}
\begin{subfigure}[b]{0.45\textwidth}
    \centering
    \pic[1.25]{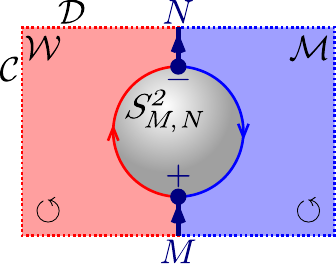}
    \caption{}
    \label{fig:lines_on_def_bdrylab}
\end{subfigure}
\caption{
(a) Point defect labelled by a morphism $f\in\mcM(M,N)$.
(b) Removing a small neighbourhood of the defect point produces a boundary component $S^2_{M,N}$.
}
\label{fig:lines_on_def}
\end{figure}
Before proceeding, let us focus for a moment on possible labels for point insertions on a surface defect.
In particular, consider a point insertion separating two lines between two surface defects, one labelled with $\mcW$ (as a module category over itself), and the other by a $\mcW$-module category $\mcM$.
By the previous section, the lines are labelled by module functors $\mcW \ra \mcM$ (equivalently, objects $M, N\in\mcM$ which correspond to module functors $-\triangleright M, -\triangleright N \colon \mcW \ra \mcM$). The point insertion is labelled by a natural transformation between the module functors (equivalently, a morphism in $f\colon M\ra N$ which corresponds to the natural transformation $\{\id_W \triangleright f\}_{W\in\mcW}$), see Figure~\ref{fig:lines_on_def_morphlab}.
Another way of labelling such point insertions is the point defect completion mentioned in Section~\ref{subsec:constr_of_defect_TQFT}:
Remove a small open ball surrounding the point in question.
It leaves a boundary component which is a stratified $2$-sphere $S^2_{M,N}$, to which the defect TQFT $Z$ assigns a vector space $Z(S^2_{M,N})$ (see Figure~\ref{fig:lines_on_def_bdrylab}).
The point insertions can then be labelled by vectors in this vector space.
The assumption we make on $Z$ is that these sets of labels are the same in the sense that the map
\begin{equation}
\label{eq:S2MN_sos}
\begin{array}{ccc}
\mcM(M,N) & \ra & Z(S^2_{M,N})\\
f & \mapsto & Z \left( \pic[1.25]{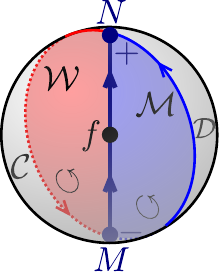} \right)
\end{array},
\end{equation}
is an isomorphism of vector spaces. Here the image of $f\in\mcM(M,N)$ is to be understood as follows:
The picture in the argument of $Z$ represents a stratified solid ball, seen as a bordism $\varnothing\ra S^2_{M,N}$.
Consequently, evaluation gives a linear map $\opC \ra Z(S^2_{M,N})$ whose image of $1\in\opC$ produces a vector in $Z(S^2_{M,N})$.
We remark that in the case $\mcC = \mcD = \mcM = \mcW$, i.e.\ when there is no surface defect, the map is indeed an isomorphism and one recovers the state space that the RT TQFT assigns to a $2$-sphere with two punctures.

\medskip

\begin{figure}
\captionsetup{format=plain, indention=0.5cm}
\centering
\begin{subfigure}[b]{0.45\textwidth}
    \centering
    \pic[1.25]{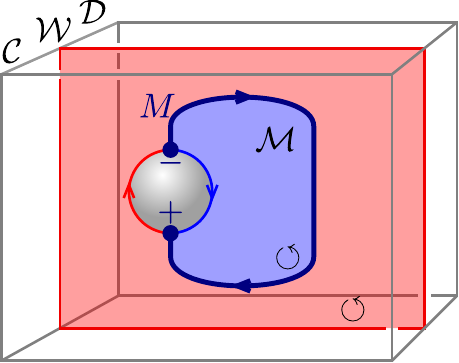}
    \caption{}
    \label{fig:B3M_bord}
\end{subfigure}
\begin{subfigure}[b]{0.45\textwidth}
    \centering
    \pic[1.25]{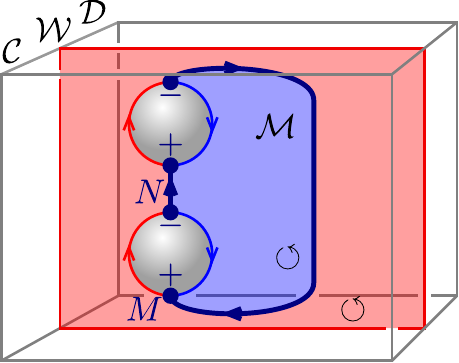}
    \caption{}
    \label{fig:PMN_bord}
\end{subfigure}
\caption{
(a) The manifold $B^3_M$: The boundary of the cube represents a single point in $S^3 \simeq \opR^3 \cup \{\infty\}$, the rest of which corresponds to the interior of the cube.
The sphere in the middle represents the incoming boundary component $S^2_{M,M}$.
(b) The manifold $P_{M,N}$: Similar to (a), but has two boundary components $S^2_{M,N}$ and $S^2_{N,M}$. Note that as a stratified $3$-manifold it is isomorphic to the cylinder $S^2_{M,N} \times [0,1]$.
}
\label{fig:bordisms_from_S3}
\end{figure}
We now define a module trace $\Theta$ on a $\mcW$-module category $\mcM$ describing a surface defect as follows:
For $M\in\mcM$, let $B^3_M$ be the stratified sphere $S^3$ with a removed open ball as in Figure~\ref{fig:B3M_bord}.
It has a single boundary component $S^2_{M,M}$, which we assume to be an incoming boundary.
Evaluating with the TQFT we obtain the linear map
\begin{equation}
Z(B^3_{M})\colon Z(S^2_{M,M}) \ra Z(\varnothing) ~,
\end{equation}
which by precomposing with \eqref{eq:S2MN_sos} can be seen as a map $\End_\mcM M \ra \opC$.
\begin{prp}
\label{prp:mod_trace}
The collection of maps $\Theta_M := Z(B^3_M)$, $M\in\mcM$ is a module trace on $\mcM$.
\end{prp}
\begin{proof}
The properties of a module trace can be shown by using the fact that upon evaluation with $Z$ only the topological configuration of defects is important:
\begin{enumerate}[i)]
\item
A simple deformation yields a homeomorphism of stratified manifolds
\begin{equation}
\hspace{-1cm}
\pic[1.25]{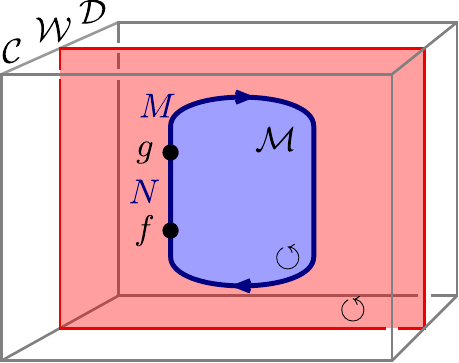} \simeq
\pic[1.25]{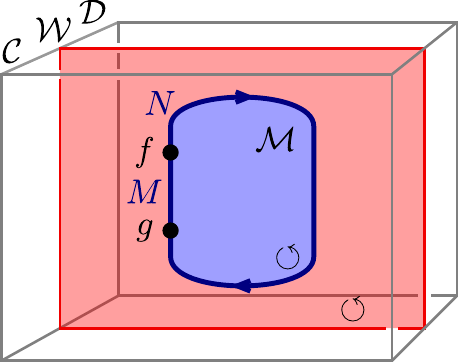} ~.
\end{equation}
By using \eqref{eq:S2MN_sos} on morphisms $g \circ f$ and $f \circ g$ one gets
$\Theta_M(g \circ f) = \Theta_N(f \circ g)$.
\item 
To show that the pairing $\omega_{M,N}$ is non-degenerate, it is enough to provide a copairing
\begin{equation}
\Omega_{M,N} \colon \opC \ra \mcM(N, M) \otimes_\opC \mcM(M, N),
\end{equation}
such that
\begin{equation}
\label{eq:inv_pairing}
\begin{array}{c}
( \omega_{M,N} \otimes_\opC \id_{M,N} ) \circ (\id_{M,N} \otimes_\opC \Omega_{M,N}) = \id_{M,N} ~, \\
(\id_{N,M} \otimes_\opC \omega_{M,N}) \circ (\Omega_{M,N} \otimes_\opC \id_{N,M}) = \id_{N,M} ~.
\end{array}
\end{equation}
Let $P_{M,N}$ be the stratified manifold as in Figure~\ref{fig:PMN_bord}.
Interpreting both its boundary components as incoming, one gets a bordism
$P_\omega\colon S^2_{M,N} \sqcup S^2_{N,M} \ra \varnothing$.
Together with the identification in \eqref{eq:S2MN_sos}, the evaluation with the TQFT $Z(P_\omega)$ gives a pairing $\mcM(N, M) \otimes_\opC \mcM(M, N) \ra \opC$ which by \eqref{eq:mod_trace_def_omega} and the definition of $\Theta$ is equal to $\omega_{M,N}$.
Similarly, interpreting the boundary components of $P_{M,N}$ as outgoing one gets a bordism $P_\Omega\colon \varnothing \ra S^2_{N,M} \sqcup S^2_{M,N}$.
We  define $\Omega_{M,N} := Z(P_\Omega)$.
It remains to show that the identities \eqref{eq:inv_pairing} hold.
They follow from the functoriality of the TQFT $Z$ and homeomorphisms
\begin{equation}
P_\omega \cup_{S^2_{M,N}} P_\Omega \simeq S^2_{N,M} \times [0,1] \, , \quad
P_\omega \cup_{S^2_{N,M}} P_\Omega \simeq S^2_{M,N} \times [0,1] \, ,
\end{equation}
i.e.\ gluing two copies of $P_{M,N}$ across a suitable boundary component gives a cylinder.
\item
This follows from the homeomorphism of stratified manifolds obtained by moving the $W$-line over the point at infinity in $S^3$:
\begin{equation}
\label{eq:mod_trace_cond_3_proof}
\hspace{-1cm}
\pic[1.25]{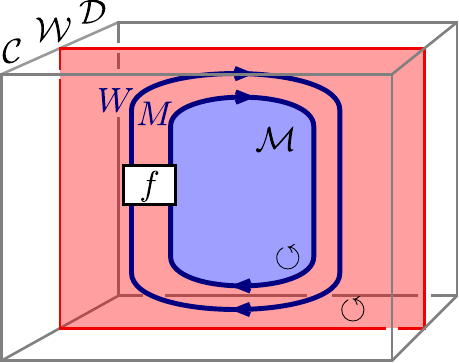} \simeq
\pic[1.25]{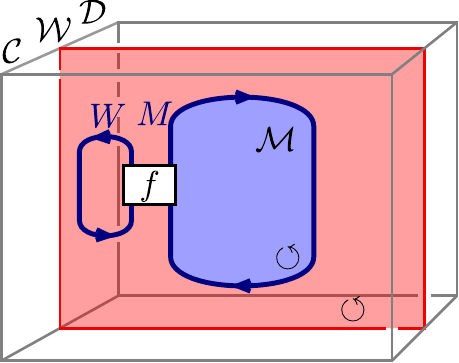} \, .
\end{equation}
\end{enumerate}
\end{proof}
A noteworthy consequence of this result is that the fusion category $\mcW$ is spherical. Indeed, $\mcW$ as an indecomposable module category over itself has a module trace given by the left categorical trace, which is then proportional to the one in Proposition~\ref{prp:mod_trace} due to Proposition~\ref{prp:mod_trace_unique}.
The homeomorphism \eqref{eq:mod_trace_cond_3_proof} for $\mcM=\mcW$ and $M=\opid_\mcW$ yields
\begin{equation}
\Theta_W f = \Theta_{W^*} f^* ~, \quad
W\in\mcW ~,~ f\in\End_\mcW W ~.
\end{equation}

\section{Comparison of the two descriptions}
\label{sec:4}

\subsection{Witt equivalence of modular categories}
\label{subsec:Witt_eq}
\begin{defn}
\label{def:Witt_eq}
Two MFCs $\mcC$ and $\mcD$ are \textit{Witt equivalent} if there exists a spherical fusion category $\mcS$ and a ribbon equivalence $\mcC\boxtimes\widetilde{\mcD}\simeq\mcZ(\mcS)$, which we call a \textit{Witt trivialisation}.
\end{defn}
\begin{rem}
\label{rem:non-deg_vs_modular_Witt_equiv}
\begin{enumerate}[i), wide, labelwidth=0pt, labelindent=0pt]
\item
The notion of Witt equivalence was introduced in \cite{DMNO} for \textit{non-degenerate braided} fusion categories, i.e.\ without an assigned ribbon structure.
There are hence two Witt groups: that of modular fusion categories and that of non-degenerate braided fusion categories.
For the application in this paper, the version with ribbon structure is the relevant one.
\item
Witt equivalence is indeed an equivalence relation on modular categories.
The set of equivalence classes forms the so-called \textit{Witt group} whose multiplication is induced by the Deligne product, the unit is given by the class consisting of Drinfeld centres and the inverses are given by braiding reversal due to existence of the equivalence \eqref{eq:CC_to_ZC_eq}, see \cite{DMNO}.
\end{enumerate}
\end{rem}

As was already mentioned in Section~ \ref{subsec:domain_walls_btw_RT_theories}, the notion of Witt equivalence turns out to be of central importance in the analysis of surface defects in $3$-dimensional TQFTs.
The following characterisation of Witt equivalence is formulated in \cite[Prop.\,5.15]{DMNO} for non-degenerate braided fusion categories; we recall the proof to show that the argument applies in the ribbon case:
\begin{prp}
\label{prp:Witt_eq_via_algs}
Two MFCs $\mcD$, $\mcE$ are Witt equivalent iff there exists a modular fusion category $\mcC$ and two commutative haploid symmetric $\D$-separable Frobenius algebras $A,B\in\mcC$ such that $\mcD \simeq \mcC_A^\loc$ and $\mcE \simeq \mcC_B^\loc$ as ribbon fusion categories.
\end{prp}
\begin{proof}
Having a Witt trivialisation $\mcD\boxtimes\widetilde{\mcE}\xra{\sim}\mcZ(\mcS)$ for some spherical fusion category $\mcS$, one can take the Deligne product of both sides with $\mcE$ and use the equivalence \eqref{eq:CC_to_ZC_eq} to get a ribbon equivalence $F\colon\mcD \boxtimes \mcZ(\mcE) \ra \mcE \boxtimes \mcZ(\mcS)$ and then set $\mcC := \mcE \boxtimes \mcZ(\mcS)$.
Let $B'\in\mcZ(\mcS)$ be the Lagrangian algebra in Example~\ref{ex:lag-alg}.
One sets $B = \opid_\mcE \boxtimes B'$.
Similarly, one picks a Lagrangian algebra $A'$ in $\mcZ(\mcE)$ and sets $A = F(\opid_\mcD \boxtimes A')$.
\end{proof}

\medskip

For the rest of the section, let $\mcC$, $A$, $B$ be as in the proposition above.
We will look for an explicit Witt trivialisation of $\mcC_A^\loc \boxtimes \widetilde{\mcC_B^\loc}$.
Let us consider the semisimple category $\mcACB$ with $A$-$B$-bimodules in $\mcC$ as objects and bimodule maps as morphisms.
We equip it with the following monoidal product: for each $M,N\in\mcACB$ we set
\begin{equation}
\label{eq:prod_over_AB}
M {{}_A \otimes_B} N := \im \pic[1.25]{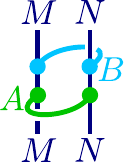} \, .
\end{equation}
Note that the monoidal unit $\opid_{\mcACB} := A \otimes B$ is in general not a simple object and
$\mcACB$ is therefore a multifusion category.

The category $\mcACB$ has a natural pivotal structure with the (co)evaluation morphisms for each $M\in\mcACB$ being
\begin{equation}
\ev_M = \pic[1.25]{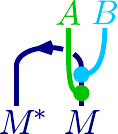}, \,
\coev_M = \pic[1.25]{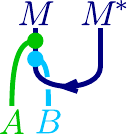}, \,
\evt_M = \pic[1.25]{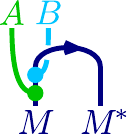}, \,
\coevt_M = \pic[1.25]{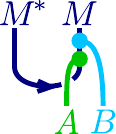}.
\end{equation}
\begin{rem}
\begin{enumerate}[i)]
\item
In the case $A=B$, the category $\mcACA$ has $A$-$A$-bimodules as objects but its tensor product is \textit{not} the tensor produce over $A$.
Since the latter is the more usual tensor product on $\mcACA$, we stress this point in order to avoid confusion.
\item
$\mcACB$ is equivalent to the category $\mcC_{A\otimes B}$ of right modules of $A\otimes B$ in $\mcC$.
In general, to define a tensor product of right modules of a $\D$-separable Frobenius algebra it needs a half-braiding with respect to which it is commutative.
In the case of the algebra $A\otimes B\in\mcC$ it is given by the ``dolphin'' half-braiding \eqref{eq:dolphin_half-br}. Note that $A\otimes B$ is in general not commutative with respect to the braiding of $\mcC$.
\end{enumerate}
\end{rem}

The multifusion category $\mcACB$ is defined in such a way that algebras and their bimodules in it correspond to algebras and their bimodules in $\mcC$ over $(A,B)$ in the sense Definitions~\ref{def:Frob_algs_over_AB} and \ref{def:bimodules_over_AB}.
In particular we have 
(see also Definitions~\ref{def:FrobAlg-over-AB} and \ref{def:Modtr_FrobalgCD}):

\begin{lem}
\label{lem:ABalgs_are_CAB}
There is an equivalence of bicategories $\FrobAlg^{\operatorname{s},\D}_{\mcC,A,B} \simeq \FrobAlg^{\operatorname{s},\D}_{\mcACB}$.
\end{lem}
\begin{proof}
Let $F\in\mcACB$ be an algebra.
Then $F$ is also an object in $\mcC$ and it is equipped with the multiplication given by $[F \otimes F \to F {{}_A \otimes_B} F \to F]$, where the first morphism is the projection to the tensor product in $\mcACB$ and the second morphism is the multiplication of $F$ in $\mcACB$. This multiplication in $\mcC$  satisfies the relations in Definition~\ref{def:Frob_algs_over_AB} because of how the tensor product in $\mcACB$ is defined (see \eqref{eq:prod_over_AB}).
One then takes the morphism $[\opid_\mcC \xra{\eta_A \otimes \eta_B} A \otimes B \xra{\eta_F} F]$ as the unit, where $\eta_A$, $\eta_B$ are the units of $A$ and $B$ and $\eta_F$ is the unit of $F$ in $\mcACB$.
To compare the bimodules it is enough to notice that the tensor product in $\mcACB$ ensures that the two $A$- and the two $B$-actions as in \eqref{eq:AB-action-from-F-action-1}, \eqref{eq:AB-action-from-F-action-2} coincide.
A similar argument applies to (symmetric, $\D$-separable) Frobenius algebras as well.
\end{proof}

\begin{prp}
\label{prp:CACB_ZACB_equiv}
The functor
\begin{equation}
\label{eq:CACB_ZACB_funct}
\mcC_A^\loc \boxtimes \widetilde{\mcC_B^\loc} \ra \mcZ(\mcACB), \qquad
M \boxtimes N \mapsto (M \otimes N, ~\g^\dol_{M,N}),
\end{equation}
is a ribbon equivalence.
Here $\g^\dol_{M,N}$ is an analogue of the ``dolphin'' half-braiding, defined for all $K\in\mcACB$ by
\begin{equation}
\g^\dol_{M,N,K} := \pic[1.25]{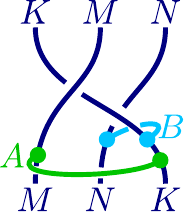} ~.
\end{equation}
\end{prp}
\begin{proof}
First we show that $\mcC_A^\loc \boxtimes \widetilde{\mcC_B^\loc} \simeq \mcZ(\mcACB)$ as braided multifusion categories.
Indeed, one has obvious equivalences of braided categories
\begin{equation}
\label{eq:equivs_chain}
\mcC_A^\loc \boxtimes \widetilde{\mcC_B^\loc} \simeq
(\mcC \boxtimes \widetilde{\mcC})_{A \boxtimes B}^\loc \simeq
\mcZ(\mcC)^\loc_{A\otimes B}~.
\end{equation}
It is known (see \cite[Cor.\,4.5]{Sch}) that for any monoidal category $\mcM$ and a commutative algebra $(C,\g)\in\mcZ(\mcM)$, one has a braided equivalence $\mcZ(\mcM)^\loc_{(C,\g)} \simeq \mcZ(\mcM_C)$, where one uses the half-braiding $\g$ of $C$ to define the monoidal product in $\mcM_C$.
Applying this to the right hand side of \eqref{eq:equivs_chain} one obtains the result.

\smallskip
\noindent
Next we show that the explicit functor \eqref{eq:CACB_ZACB_funct} is a braided equivalence.
First, using that $M$ and $N$ are local modules one can check that $\g^\dol_{M,N,K}$ does indeed satisfy the hexagon identity.
Next, by the argument above, $\mcZ(\mcACB)$ is fusion and has the same Frobenius-Perron dimension as $\mcC_A^\loc \boxtimes \widetilde{\mcC_B^\loc}$.
The functor \eqref{eq:CACB_ZACB_funct} is then a braided functor between non-degenerate braided fusion categories and therefore fully faithful by Proposition~\ref{prp:non-deg_fully-faith}.
Hence by Proposition~\ref{prp:equiv_ito_FPdims} it is an equivalence.

\smallskip
\noindent
Finally, the explicit equivalence \eqref{eq:CACB_ZACB_funct} implies in particular that $\mcZ(\mcACB)$ is spherical: it is enough to check that the left and the right dimensions of simple objects coincide (see Proposition~\ref{prp:spherical_ito_simples}). We know that all simple objects are of the form $\mu\otimes\nu$ for $\mu\in\Irr_{\mcC_A^\loc}$, $\nu\in\Irr_{\mcC_B^\loc}$ for which the left/right dimensions in $\mcZ(\mcACB)$ are the product of those of $\mu$ and $\nu$ and are hence equal.
The category $\mcZ(\mcACB)$ is therefore also ribbon (see Proposition~\ref{prp:ribb_iff_spherical}).
The equivalence \eqref{eq:CACB_ZACB_funct} can also be checked to preserve braidings and twists and is therefore a ribbon equivalence.
\end{proof}

\begin{rem}
\label{rem:ACB_not_spherical}
The category $\mcACB$ need not be spherical, even though its Drinfeld centre $\mcZ(\mcACB)$ is spherical, as shown in the proof above.
Indeed, the left and the right traces of $f\in\End_{\mcACB}(M)$ in general need not be equal:
\begin{equation}
\tr_l f = 
\pic[1.25]{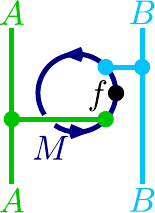} =
\pic[1.25]{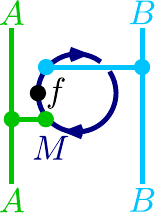} \neq
\pic[1.25]{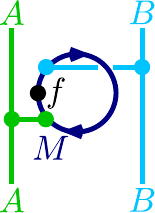} =
\tr_r f \, .
\end{equation}
\end{rem}

Proposition~\ref{prp:CACB_ZACB_equiv} in particular implies that $\mcZ(\mcACB)$ is fusion and hence that $\mcACB$ is indecomposable.
Proposition~\ref{prp:multifusion_center} then in turn implies the existence of a braided equivalence $\mcZ(\mcACB) \simeq \mcZ(\mcF)$ where 
\begin{equation}
\mcF :=(\mcACB)_{ii}
\end{equation}
any component of $\mcACB$, and which is then automatically fusion. An explicit equivalence is given by
\begin{equation}
\label{eq:ZCAB_to_ZF_funct}
\mcZ(\mcACB) \ra \mcZ(\mcF), \quad M \mapsto \opid_\mcF \, {{}_A \otimes_B} \, M \, .
\end{equation}
Indeed, we show in the Appendix~\ref{app:Witttriv_ribbon_funct} that \eqref{eq:ZCAB_to_ZF_funct} is in fact a ribbon functor.
Then, since $\mcZ(\mcACB)$ and $\mcZ(\mcF)$ are non-degenerate and have equal Frobenius-Perron dimensions, by Propositions~\ref{prp:equiv_ito_FPdims} and \ref{prp:non-deg_fully-faith} it is an equivalence.

The composition of the functors \eqref{eq:CACB_ZACB_funct} and \eqref{eq:ZCAB_to_ZF_funct} is then a Witt trivialisation for $\mcC_A^\loc\boxtimes\widetilde{\mcC_B^\loc}$.

\subsection{Equivalence of the two descriptions of domain walls}
Let $\mcC$ be a MFC and $A,B\in\mcC$ two haploid commutative symmetric $\D$-separable Frobenius algebras.
We have looked at two ways to describe domain walls separating two theories of RT type labelled by MFCs $\mcC_A^\loc$ and $\mcC_B^\loc$.
The first one comes from the construction of the defect TQFT $Z_\mcC$ defined in Section~\ref{subsec:constr_of_defect_TQFT}, while the second one is a physics inspired argument from \cite{FSV}, which was summarised in Section~\ref{subsec:domain_walls_btw_RT_theories} and expanded upon in Section~\ref{subsec:mod_traces_from_spheres}.
Below we argue that both descriptions are essentially the same in the sense that the resulting bicategories of surface defects are equivalent.

\medskip

We know from Section~\ref{subsec:bicats_of_domain_walls} and Lemma~\ref{lem:ABalgs_are_CAB} that the bicategory of domain walls obtained from the defect TQFT $Z_\mcC$ is equivalent to $\FrobAlg^{\operatorname{s},\D}_{\mcACB}$.
Following Section~\ref{subsec:Witt_eq}, let $\mcF$ be a component category of the multifusion category $\mcACB$ as in Proposition~\ref{prp:multifusion_cats}.
The existence of the Witt trivialisation $\mcC_A^\loc \boxtimes \widetilde{\mcC_B^\loc} \xra{\sim} \mcZ(\mcF)$ implies that the bicategory of domain walls according to \cite{FSV} is $\mcF$-$\operatorname{Mod}^{\tr}$.
We have:
\begin{thm}\label{thm:bicat-equiv}
One has the following commutative diagram of equivalences, inclusion and forgetful functors between bicategories:
\begin{equation}
\label{eq:bicategories_diagram}
\begin{tikzcd}
\mcF\operatorname{-Mod}^{\tr}                         \arrow[hookrightarrow]{d}{\text{include}}
& \operatorname{FrobAlg}_\mcF^{\operatorname{s},\D}   \arrow{d}{\text{forget}} \arrow{l}{\text{iii)}}[swap]{\sim}\arrow{r}{\sim}[swap]{\text{iv)}} 
& \operatorname{FrobAlg}_\mcACB^{\operatorname{s},\D} \arrow{d}{\text{forget}}\\
\mcF\operatorname{-Mod}     
& \operatorname{Alg}_\mcF   \arrow{l}{\text{i)}}[swap]{\sim}\arrow{r}{\sim}[swap]{\text{ii)}}
& \operatorname{Alg}_\mcACB
\end{tikzcd}.
\end{equation}
\end{thm}
\begin{proof}
We exploit various relations between bicategories introduced in Sections~\ref{subsec:module_cats} and \ref{subsec:Witt_eq}.

\smallskip
\noindent
Equivalence i) in \eqref{eq:bicategories_diagram} is given by the Proposition~\ref{prp:Alg_Mod_equiv}, and equivalence iii) follows from Proposition~\ref{prp:FrobalgCD_equiv_Modtr}. It is clear that the left square commutes (on the nose).

\smallskip
\noindent
Equivalence ii) is given by inclusion.
Indeed, for any indecomposable multifusion category $\mcA$ and its component category $\mcA_{ii}$, an algebra $A\in\mcA_{ii}$ is automatically an algebra in $\mcA$ by providing it with the unit $[\opid \twoheadrightarrow\opid_i\xra{\eta} A]$ where $\opid_i\in\mcA_{ii}\subset\mcA$ is the restriction of the tensor unit.
The inclusion is fully faithful, meaning that upon inclusion of two algebras $A,B\in\mcA_{ii}$, the corresponding bimodule categories ${}_A (\mcA_{ii})_B$ and ${}_A \mcA_B$ are equivalent.
Moreover, any simple algebra in $\mcA$ is Morita equivalent to one in $\mcA_{ii}$ (see e.g.\ \cite[Rem.\,3.9]{KZ2}), which implies the essential surjectivity of the inclusion.
It is easy to check that in case $\mcA$ is pivotal, the inclusion preserves the structure of a symmetric $\D$-separable Frobenius algebra, which then gives equivalence iv), as well as commutativity of the right square (again on the nose).
\end{proof}

\appendix
\section{Appendix: Witt trivialisation as ribbon functor}
\label{app:Witttriv_ribbon_funct}
Let $\mcC$ be a modular fusion category and $A,B\in\mcC$ two haploid commutative symmetric $\D$-separable Frobenius algebras.
We show here that the component category $\mcF := (\mcACB)_{ii}$ is spherical and that the functor \eqref{eq:ZCAB_to_ZF_funct} is ribbon.

\medskip

From Proposition~\ref{prp:multifusion_cats} we know that one has a decomposition $A\otimes B \cong \bigoplus_i F_i$ where $F_i\in\mcACB$ are simple and mutually non-isomorphic.
Each $F_i$ is canonically a symmetric $\D$-separable Frobenius algebra in $\mcACB$ and the decomposition of $A\otimes B$ is an isomorphism of algebras in $\mcACB$.
Objects of $\mcF$ can therefore be seen as objects of $\mcACB$ such that only the action of $F:=F_i$ is non-trivial.
$\mcF$, as a pivotal category, is hence equivalent to $F$-$F$-bimodules in $\mcACB$. By Lemma~\ref{lem:ABalgs_are_CAB}, $F$ is automatically a symmetric $\D$-separable Frobenius algebra over $(A,B)$ in $\mcC$ and $\mcF$ is in turn equivalent to $F$-$F$-bimodules over $(A,B)$ in $\mcC$.

\medskip

The key observation now is that $F$ is \textit{haploid} in $\mcC$, i.e.\ $\dim \mcC(\opid, F) = 1$.
Indeed, according to Proposition~\ref{prp:multifusion_cats} one has $\dim \mcACB(A\otimes B, F) = 1$ and one can check that the induction/forgetful functors
\begin{equation}
\begin{array}{rcl}
\Ind\colon\mcC & \ra     & \mcACB  \\
      X   & \mapsto & A \otimes X \otimes B 
\end{array}
\quad , \quad
\begin{array}{rcl}
U\colon\mcACB & \ra     & \mcC  \\
      M   & \mapsto & M 
\end{array}
\end{equation}
are biadjoint to each other.
For a morphism $[f\colon M\ra M]\in\mcF$ we now have:
\begin{align}
\nonumber
\tr_l f &\overset{(1)} =        \pic[1.25]{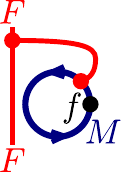}
\overset{(2)}= \frac{1}{\dim F} \pic[1.25]{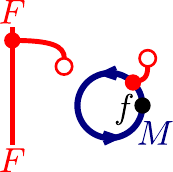}
= \frac{1}{\dim F}              \pic[1.25]{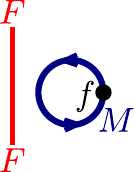}\\
&
\overset{(3)}= \frac{1}{\dim F} \pic[1.25]{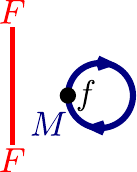}
=                               \pic[1.25]{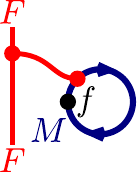}
= \tr_r f ~.
\end{align}
Here, in step (1) one applies the definition of a left trace \eqref{eq:left_right_traces} (as an endomorphism of the tensor unit, i.e.\ of $F$) to the category of $F$-$F$-bimodules over $(A,B)$.
For step (2) one notes that $\eta$, $(\dim F)^{-1}\vareps$
can be taken for the inclusion/projection morphisms $\opid \ra F$ and $F \ra \opid$, see the proof of Lemma~\ref{lem:dimA-nonzero}.
In (3) we use that $\mcC$ is ribbon, hence spherical.
$\mcF$ is therefore a spherical fusion category.

\medskip

It remains to show that the functor \eqref{eq:ZCAB_to_ZF_funct} is ribbon.
Having objects $(M,\gamma) \in \mcZ(\mcACB)$ and $N\in\mcACB$ we note that because of the decomposition $A\otimes B \cong \bigoplus_i F_i$ one can write the half-braiding $\gamma_N$ as a sum
\begin{equation}
\left[ \gamma_N\colon M {{}_A \otimes_B} N \ra N {{}_A \otimes_B} M \right] =
\sum_{i,j} \left[\gamma_N^{ij}\colon M \otimes_{F_i} N \ra N \otimes_{F_j} M \right] ~.
\end{equation}
In particular, if $N$ is in $\mcF$, the sum has a single term $\gamma_N^{ii}$, i.e.\ the half-braiding restricts to the component in $\mcF$.
Since $\gamma$ consists of isomorphisms, we see that $M$ has only diagonal components in the decomposition $\mcACB \simeq \bigoplus_{ij} F_i \, {_A \otimes_B} \, \mcACB \, {_A \otimes_B} \, F_j$.
Therefore, braidings and twists in $\mcZ(\mcF)$ are simply projections of braidings and twists in $\mcZ(\mcACB)$ onto the component $\mcF$ and the functor \eqref{eq:ZCAB_to_ZF_funct} is precisely this projection.

\newpage

\newcommand{\arxiv}[2]{\href{http://arXiv.org/abs/#1}{#2}}
\newcommand{\doi}[2]{\href{http://doi.org/#1}{#2}}

\end{document}